\documentclass[11pt]{article}

\usepackage{lmodern}
\usepackage[T1]{fontenc}
\usepackage[utf8]{inputenc}
\usepackage[english]{babel}
\usepackage{amsthm}
\usepackage{amsmath}
\usepackage{amssymb}
\usepackage{mathtools}
\usepackage{graphicx}
\usepackage{enumitem}
\usepackage{subcaption}
\usepackage{stmaryrd}
\usepackage{dsfont}
\usepackage{geometry}
\usepackage[hidelinks]{hyperref}
\usepackage{authblk}
\usepackage{fancyhdr}
\usepackage[justification=centering]{caption}
\usepackage{float}
\newcommand{\R}{\mathbb{R}}
\newcommand{\N}{\mathbb{N}}
\newcommand{\E}{\mathbb{E}}

\newcommand{\PP}{\mathbb{P}}
\newcommand{\ind}{\mathds{1}}
\newcommand{\dom}{d_{\mathrm{OM}}}
\newcommand{\dH}{d_{\mathrm{H}}}

\newcommand{\gin}{\Delta_{\mathrm{in}}}
\newcommand{\gout}{\Delta_{\mathrm{out}}}
\newcommand{\sub}{\mathrm{sub}}
\newcommand{\ins}{\mathrm{ins}}
\newcommand{\del}{\mathrm{del}}
\newcommand{\Wd}{W_{\bar{d}}}
\newcommand{\tmix}{\tau_{\mathrm{mix}}}

\newtheorem{theorem}{Theorem}[section]
\newtheorem{proposition}[theorem]{Proposition}
\newtheorem{lemma}[theorem]{Lemma}

\theoremstyle{definition}
\newtheorem{definition}[theorem]{Definition}
\newtheorem{assumption}{Assumption}

\theoremstyle{remark}
\newtheorem{remark}[theorem]{Remark}

\author[1]{Ottavio Khalifa}
\affil[1]{Universit\'e Paris Cit\'e, Universit\'e Sorbonne Paris Nord, INSERM, INRAE,\\
Centre for Research in Epidemiology and StatisticS (CRESS), Paris, France\\
\texttt{ottavio.khalifa@inserm.fr}}

\title{\bf Consistency of Optimal Matching-based Clustering\\
for Mixtures of Markov Chains}
\date{\today}

\begin{document}
\maketitle

\begin{abstract}
We study clustering of categorical sequences using the Optimal Matching (OM) distance under finite mixtures of finite-state Markov chains. We show that the normalized OM distance between two independent chains converges almost surely to a deterministic population quantity, concentrates exponentially around its finite-horizon mean, and admits an $O(\sqrt{\log n/n})$ convergence rate when the two chains have the same transition kernel. These population quantities yield a natural separation condition: the largest within-component limit must be smaller than the smallest between-component limit. Under this condition, hierarchical clustering with any bracketed linkage and Partitioning Around Medoids consistently recover the latent mixture partition. We also propose a consistent estimator of the number of components based on empirical OM distance profiles. The results extend to finite-state hidden Markov models and multichannel categorical observations. Overall, they provide a statistical justification for standard OM-based clustering methods for categorical time series.

\end{abstract}

\medskip
\noindent\textbf{Keywords:} Optimal Matching, edit distance, categorical sequences, sequence analysis, clustering consistency, hierarchical clustering, $K$-medoids, Markov chains, hidden Markov models.

\medskip
\noindent\textbf{MSC2020 subject classifications:} Primary 62H30; secondary 60J10, 62M05.

\section{Introduction}
\label{sec:intro}

Categorical sequences arise in many applications, including employment histories, family trajectories, and care pathways. A standard approach to analyze such data is to compute pairwise dissimilarities between trajectories and apply a clustering algorithm to the resulting dissimilarity matrix. A recent review of clustering methods for categorical sequences \cite{khalifa2026clustering} shows that the dominant approach combines the Optimal Matching (OM) dissimilarity with hierarchical clustering or $K$-medoids, while Markov chains and hidden Markov models are the most widely used probabilistic models for such data. The use of OM for clustering categorical sequences was notably imported from bioinformatics into the social sciences by \cite{abbott1986}. OM is a weighted edit, or global-alignment, dissimilarity, closely related to the classical Levenshtein distance \cite{levenshtein1966,Bioinformatics}. It is implemented in the TraMineR package \cite{traminer} and is widely used in life-course research, demography, and epidemiology \cite{SASystematicreview, SAPPF}.

For an arbitrary cost scheme, OM is only a dissimilarity. Under the metric cost scheme of Assumption~\ref{ass:metric}, which we assume throughout the paper, OM is a genuine distance on the set of finite sequences, so we will refer to it as the OM distance from now on. 

The methodological literature on OM in biology and social science has largely focused on the OM distance itself: how to compute it efficiently \cite{needleman1970}, which features of a trajectory it captures, and how substitution and insertion/deletion costs should be chosen \cite{OMReview}. At the same time, a probabilistic literature has studied the asymptotic behavior of edit and global-alignment scores under several stochastic models \cite{lember2018,bilardi2022computable,houdre2019}. Yet much less is known about what these results imply when the OM distance is used for clustering categorical sequences. In particular, what population quantity does the normalized OM distance estimate? If the observed sequences are generated by a mixture of stochastic processes, under what conditions does OM-based clustering recover the mixture components? And how does trajectory length affect the reliability of recovery?

\textbf{Contributions.} We consider a finite mixture of $K$ irreducible and aperiodic Markov chains on a finite alphabet $\Sigma$. We observe a dataset $(X_i)_{1\le i\le N}$ of $N$ independent sequences of length $n$, and let $Z_i\in\{1,\ldots,K\}$ denote the latent mixture component generating sequence $X_i$. We study clustering procedures that use only their pairwise OM distances. Our main results on the recovery of the latent mixture partition cover hierarchical agglomerative clustering with bracketed linkages---including single, complete, and average linkage---and Partitioning Around Medoids (PAM) run to one-swap stationarity (Theorems~\ref{thm:SL} and~\ref{thm:kmedoids}).

To this end, for two independent Markov chains $X=(X_t)_{t\ge1}$ and $Y=(Y_t)_{t\ge1}$ on $\Sigma$, with transition kernels $P$ and $Q$, respectively, we use the normalized OM distance

\[
\hat\gamma_n(X,Y)
:=
\frac{1}{n}\dom(X_{1:n},Y_{1:n}).
\]

We establish that $\hat\gamma_n$ converges almost surely to a deterministic quantity $\gamma(P,Q)$ that depends only on the transition kernels and not on the initial distributions (Proposition~\ref{prop:convergence}). We obtain exponential concentration of $\hat\gamma_n$ around its finite-horizon mean (Proposition~\ref{prop:concentration}). We prove, for general metric OM costs, an $O(\sqrt{\log n/n})$ convergence rate for its finite-horizon mean when $P=Q$ (Proposition~\ref{prop:diagonal-rate}). We also derive computable lower and upper bounds on $\gamma(P,Q)$ depending only on the stationary distributions of the two chains (Proposition~\ref{prop:bounds}).

Hence, on a dataset $(X_i)_{1 \le i \le N}$ of $N$ sequences generated by a mixture of $K$ Markov chains, each entry of the pairwise normalized distance matrix $\left(\hat{\gamma}_n(X_i,X_j)\right)_{1 \le i,j \le N}$ converges to $\Gamma_{Z_iZ_j}$, where $\Gamma := \bigl(\gamma(P_k,P_{k'})\bigr)_{1\le k,k'\le K}$. This leads to a natural separation requirement (Assumption~\ref{ass:separation}): the largest within-component asymptotic OM value must be smaller than the smallest between-component one, that is,

\[
\max_k \Gamma_{kk}
<
\min_{k\neq k'} \Gamma_{kk'}.
\]

Under this condition, hierarchical clustering with single, complete and average linkages, with a dendrogram cut at $K$ clusters, consistently recover the true partition (Theorem~\ref{thm:SL} and Remark~\ref{rmk:other-linkages}). The same holds for PAM: every one-swap local optimum yields the correct clustering with probability tending to one (Theorem~\ref{thm:kmedoids}).

When $K$ is unknown, we introduce a data-driven selection rule based on OM distances, and show that it consistently estimates the number of mixture components (Theorem~\ref{thm:K-selection}). Finally, we give sketches of proofs for similar convergence, concentration, and clustering guarantees for multichannel categorical sequences generated by mixtures of finite-state hidden Markov models, with the multichannel variant of the OM distance (Proposition~\ref{prop:gen-concentration} and Theorem~\ref{thm:gen-clustering}).

\textbf{Related Work.} The asymptotic theory of global alignment scores, including the longest common subsequence (LCS) and weighted edit distances, is well established. For independent i.i.d.\ sequences, \cite{chvatal1975longest} established the convergence of the normalized LCS, and \cite{alexander1994rate} obtained a convergence rate for its mean. This rate argument was later extended to more general alignment scores by \cite{lember2012}. In a dependent setting, \cite{lember2018} study global alignment scores for pairwise Markov chains, building on the subadditive limit and deriving exponential concentration from Markov-chain bounded-difference inequalities. \cite{houdre2019} obtain a rate of convergence for the normalized expected LCS in hidden Markov models, while \cite{bilardi2022computable} study limiting expected edit distances, computable bounds, and Monte Carlo estimation in the i.i.d.\ setting. Our probabilistic analysis builds on these results and extends them to general OM distances in a form that can be used for clustering.

The clustering of dependent sequences has also been studied from different perspectives. \cite{consistent} give consistency guarantees for clustering time series generated by stationary ergodic processes, using a distributional distance between process laws. Their approach targets separation at the level of the underlying process distributions, while the present work considers separation defined through the OM distance. A separate literature studies learning and clustering finite mixtures of Markov chains using methods designed specifically for that generative model \cite{kausik2023,spaeh2023,nearoptimal}. These methods make direct use of the transition dynamics, while OM-based clustering relies only on the pairwise distance matrix.

To the best of our knowledge, no previous work has studied the consistency of clustering approaches based on the OM distance.

\textbf{Outline.} Section~\ref{sec:OM} introduces the OM distance, establishes the convergence and concentration properties of $\hat\gamma_n$ for Markov chains, gives a convergence rate when $\hat\gamma_n$ is computed from realizations of Markov chains with the same transition kernels, and derives computable bounds on $\gamma$. Section~\ref{sec:clustering} develops the clustering recovery results for hierarchical clustering and PAM, together with the data-driven selection of $K$. Section~\ref{sec:extensions} extends the main convergence, concentration, and clustering results to multichannel categorical sequences generated by mixtures of Hidden Markov Models, with the multichannel OM distance. Section~\ref{sec:numerics} reports numerical illustrations. All proofs are gathered in Section~\ref{sec:proofs}.

\section{The Optimal Matching Distance}
\label{sec:OM}

\subsection{Alphabet, Strings, and Edit Operations}

\begin{definition}[Alphabet and strings]
\label{def:alphabet}

Let $\Sigma$ be a finite alphabet with $|\Sigma|=d\ge 2$. For $n\ge 0$, write $\Sigma^n$ for the set of strings $x_{1:n}=(x_1,\ldots, x_n)$ with letters in $\Sigma$, with $\Sigma^0=\{\emptyset\}$, and set $\Sigma^\star:=\bigsqcup_{n\ge 0}\Sigma^n$. We write $\Sigma^{\N}$ for the set of one-sided infinite sequences.
\end{definition}

\begin{definition}[Optimal Matching dissimilarity]
\label{def:OM}
Fix a symmetric substitution cost $c_{\sub}:\Sigma\times\Sigma\to[0,\infty)$ with $c_{\sub}(a,a)=0$, and insertion/deletion costs $c_{\ins},c_{\del}:\Sigma\to(0,\infty)$. For $x_{1:n},y_{1:m}\in\Sigma^\star$, an \emph{edit path} from $x_{1:n}$ to $y_{1:m}$ is a finite sequence of elementary operations---substitution of a letter $a$ by a letter $b$ (cost $c_{\sub}(a,b)$), deletion of a letter $a$ (cost $c_{\del}(a)$), insertion of a letter $b$ (cost $c_{\ins}(b)$)---applied successively and transforming $x_{1:n}$ into $y_{1:m}$. The cost of an edit path $\pi$, denoted $\mathrm{cost}(\pi)$, is the sum of the costs of its operations, and the \emph{OM dissimilarity} is
\[
\dom(x_{1:n},y_{1:m}) := \min\{\mathrm{cost}(\pi):\ \pi\text{ edit path from }x_{1:n}\text{ to }y_{1:m}\}.
\]

\end{definition}

\begin{assumption}[Metric cost scheme]
\label{ass:metric}

The costs $(c_{\ins},c_{\del},c_{\sub})$ satisfy:
\begin{enumerate}[label=(\roman*)]
\item $\delta(a):=c_{\ins}(a)=c_{\del}(a)$ for all $a\in\Sigma$;
\item $c_{\sub}(a,b)=c_{\sub}(b,a)$ for all $a,b\in\Sigma$;
\item $c_{\sub}(a,c)\le c_{\sub}(a,b)+c_{\sub}(b,c)$ for all $a,b,c\in\Sigma$;
\item $c_{\sub}(a,b)\le\delta(a)+\delta(b)$ for all $a,b\in\Sigma$;
\item $\delta(a)\le c_{\sub}(a,b)+\delta(b)$ for all $a,b\in\Sigma$;
\item $c_{\sub}(a,b)>0$ for all $a\neq b$.
\end{enumerate}
\end{assumption}

Let $\Sigma_\bot:=\Sigma\cup\{\bot\}$, where $\bot\notin\Sigma$ is a formal gap symbol used to represent insertions and deletions in an alignment. We extend $c_{\sub}:\Sigma\times\Sigma\to[0,+\infty)$ by

\begin{equation}
\label{eq:dbar}
\bar d:\Sigma_\bot\times\Sigma_\bot\to\R^+,\qquad
\bar d(a,b):=\begin{cases} c_{\sub}(a,b) & a,b\in\Sigma,\\ \delta(a) & a\in\Sigma,\ b=\bot,\\ \delta(b) & a=\bot,\ b\in\Sigma,\\ 0 & a=b=\bot.\end{cases}
\end{equation}

Under Assumption~\ref{ass:metric}, $\bar d$ is a metric on $\Sigma_\bot$.

An \emph{alignment} of $(x_{1:n},y_{1:m})$ is a pair $(\tilde x_{1:L},\tilde y_{1:L})\in\Sigma_\bot^L\times\Sigma_\bot^L$ with $(\tilde x_\ell,\tilde y_\ell)\neq(\bot,\bot)$ for every $\ell$, such that deleting the symbols $\bot$ from $\tilde x$ recovers $x_{1:n}$ and deleting them from $\tilde y$ recovers $y_{1:m}$. Its cost is

\begin{equation}
\label{eq:aligncost}
\mathrm{cost}(\tilde x,\tilde y) := \sum_{\ell=1}^L \bar d(\tilde x_\ell,\tilde y_\ell).
\end{equation}

Under Assumption~\ref{ass:metric}, $\dom$ is a metric on $\Sigma^\star$ and coincides with the minimum of~\eqref{eq:aligncost} over all alignments of the two strings \cite{Bioinformatics}. We call it the OM distance from here on.

Let

\begin{equation}
\label{eq:Mdef}
M := \max_{a,b\in\Sigma}c_{\sub}(a,b),
\end{equation}

so that $\dom(x_{1:n},y_{1:n})\le Mn$, since substituting $x_i$ by $y_i$ at each position is an edit path and has total cost at most $Mn$.

\subsection{Almost Sure Convergence}

Let $X=(X_t)_{t\ge 1}$ and $Y=(Y_t)_{t\ge 1}$ be independent processes on $\Sigma$. Set

$$
D_n := \dom(X_{1:n},Y_{1:n}).
$$

The asymptotic analysis of random sequence alignment costs is classically based on subadditivity and ergodic arguments. We follow the strategy of \cite{lember2018,houdre2019} to state a convergence result in the present OM setting and give a detailed proof, including the treatment of arbitrary initial distributions.

\begin{proposition}[Convergence of the normalized OM distance]
\label{prop:convergence}
Let $X$ and $Y$ be independent time-homogeneous Markov chains on $\Sigma$ with irreducible kernels $P,Q$, stationary laws $\pi,\pi'$, and arbitrary initial laws $\mu,\mu'$. Then $\hat\gamma_n:=D_n/n$ converges almost surely to a deterministic limit $\gamma(P,Q)$, which does not depend on the initial laws and is given by

$$
\gamma(P,Q)
=
\inf_{n\ge 1}\frac{\E_{\pi,\pi'}[D_n]}{n}.
$$

\end{proposition}

\begin{remark}[Relation to Ornstein's distance]
\label{rmk:gamma-not-metric}
$\gamma$ is not a metric on laws: in general $\gamma(P,P) \ne 0$. Taking the infimum of the asymptotic OM cost over all stationary couplings of the two Markov chains would define an Ornstein-type distance on laws \cite{ornstein1973, shields1996}, but our setting forces independent pairs of trajectories, which correspond to the independent coupling.
\end{remark}

\subsection{Concentration and a Diagonal Convergence Rate for Markov Chains}

For an irreducible aperiodic Markov chain $(Z_t)_{t\ge0}$ with transition matrix $R$ and stationary distribution $\pi_R$, define
\begin{equation}
\label{eq:tmix-def}
\tmix(R)
:=
\min\left\{
t\ge1:
\max_{z}
\left\|
\PP(Z_t\mid Z_0=z)-\pi_R
\right\|_{\mathrm{TV}}
\le\frac14
\right\}.
\end{equation}

Here $\|\cdot\|_{\mathrm{TV}}$ denotes the total variation distance. The following result is a direct application of the McDiarmid-type inequality for Markov chains of \cite[Corollary~2.10 and Remark~2.11]{paulin}.

\begin{proposition}[Concentration around the finite-horizon mean]
\label{prop:concentration}
Assume Assumption~\ref{ass:metric}. Let $X$ and $Y$ be independent time-homogeneous irreducible aperiodic Markov chains on $\Sigma$, with transition matrices $P$ and $Q$ and arbitrary initial distributions. Let
\[
\tmix:=\tmix(P\otimes Q).
\]
Then, there exists a universal constant $C_{\mathrm{Pau}}>0$ such that, setting
\begin{equation}
\label{eq:Cdef}
C=C(M,P,Q)
:=
C_{\mathrm{Pau}}\,M^2\,\tmix,
\end{equation}
for every $n\ge1$ and every $s>0$,
\begin{equation}
\label{eq:concentration-centered}
\PP\left(
\left|
\hat\gamma_n-\E[\hat\gamma_n]
\right|
\ge s
\right)
\le
2\exp\left(
-\frac{2s^2n}{C}
\right).
\end{equation}
\end{proposition}

To turn Proposition~\ref{prop:concentration} into a quantitative deviation bound around the limit $\gamma(P,Q)$, one also needs to control the finite-horizon bias
\[
\frac1n\E[D_n]-\gamma(P,Q).
\]
We obtain such a rate when $P=Q$ by relating general OM costs to the LCS (Longest Common Subsequence) alignment score studied by \cite{houdre2019}. We proceed as follows.

Consider first the classical cost scheme
\[
\delta(a)\equiv1,
\qquad
c_{\sub}(a,b)
=
\begin{cases}
0, & a=b,\\
2, & a\neq b.
\end{cases}
\]
For two words of lengths $p$ and $q$, the corresponding OM distance is
\[
\dom(x_{1:p},y_{1:q})
=
p+q-2\,\mathrm{LCS}(x_{1:p},y_{1:q}),
\]
and, in particular, for equal lengths,
\[
\dom(x_{1:n},y_{1:n})
=
2\bigl(
n-\mathrm{LCS}(x_{1:n},y_{1:n})
\bigr).
\]
Thus the LCS is an alignment score associated with a particular OM cost scheme.

The next lemma shows that, for a general metric OM scheme, the analogue of the LCS satisfies all the properties needed in the rate argument of the proof of \cite[Theorem~3.1]{houdre2019}.

\begin{lemma}[Properties of the OM-induced alignment score]
\label{lem:OM-score-properties}
Assume Assumption~\ref{ass:metric}, and define
\[
s(a,b):=\delta(a)+\delta(b)-c_{\sub}(a,b),
\qquad a,b\in\Sigma.
\]
For $x_{1:p}\in\Sigma^p$ and $y_{1:q}\in\Sigma^q$, let
\[
L_s(x_{1:p},y_{1:q})
:=
\max
\left\{
\sum_{h=1}^r s(x_{i_h},y_{j_h})
:
1\le i_1<\cdots<i_r\le p,\;
1\le j_1<\cdots<j_r\le q
\right\},
\]
where $r=0$ is allowed, with the empty sum equal to zero. Let
\[
F_s:=\max_{a,b\in\Sigma}s(a,b),
\qquad
A_s:=\max_{a,a',b\in\Sigma}
|s(a,b)-s(a',b)|.
\]
Then,

\begin{enumerate}[label=(\roman*),leftmargin=*]

\item The score $s$ is non-negative, symmetric, and
\begin{equation}
\label{eq:OM-score-identity}
\dom(x_{1:p},y_{1:q})
=
\sum_{i=1}^p\delta(x_i)
+
\sum_{j=1}^q\delta(y_j)
-
L_s(x_{1:p},y_{1:q}).
\end{equation}

\item For all $p,p',q,q'\ge0$, $x_{1:p+p'}\in\Sigma^{p+p'}$, and
$y_{1:q+q'}\in\Sigma^{q+q'}$,
\[
L_s(x_{1:p+p'},y_{1:q+q'})
\ge
L_s(x_{1:p},y_{1:q})
+
L_s(x_{p+1:p+p'},y_{q+1:q+q'}).
\]
Moreover, extending either word cannot decrease $L_s$.

\item For every $x_{1:p}\in\Sigma^p$ and $y_{1:q}\in\Sigma^q$,
\[
0
\le
L_s(x_{1:p},y_{1:q})
\le
F_s\min(p,q).
\]

\item If two words differ at a single position, then changing that letter
changes $L_s$ by at most $A_s$. More precisely, if $x,x'\in\Sigma^p$ differ
at one position, then
\[
|L_s(x,y)-L_s(x',y)|
\le
A_s,
\]
and the same bound holds for a one-letter change in $y$.

\item If $X$ and $Y$ are independent stationary copies of the same process,
then for every $n\ge1$ and every $p,q\ge0$ with $p+q\le2n$,
\begin{equation}
\label{eq:OM-rectangular-block}
\E L_s(X_{1:p},Y_{1:q})
\le
\frac12\,
\E L_s(X_{1:2n},Y_{1:2n}).
\end{equation}

\end{enumerate}
\end{lemma}

Lemma~\ref{lem:OM-score-properties} provides all the ingredients needed to adapt the rate argument of \cite[Theorem~3.1]{houdre2019} to $L_s$. The independent-identical specialization is described in the proof of \cite[Corollary~3.1]{houdre2019}. The partition construction and counting argument remain the same in our setting. This yields the following quantitative bound on the finite-horizon bias when $P = Q$. 

\begin{proposition}[Convergence rate on the diagonal]
\label{prop:diagonal-rate}
Let $X$ and $Y$ be two independent time-homogeneous irreducible aperiodic Markov chains on $\Sigma$, with the same transition matrix $P$ and the same initial distribution $\mu$. Under Assumption~\ref{ass:metric}, there exists
a constant
\[
C_{\mathrm{rate}}
=
C_{\mathrm{rate}}(P,c_{\sub},\delta)>0
\]
such that, for every $n\ge2$,
\begin{equation}
\label{eq:diagonal-rate}
\left|
\frac1n\E_{\mu,\mu}[D_n]
-
\gamma(P,P)
\right|
\le
C_{\mathrm{rate}}
\sqrt{\frac{\log n}{n}}.
\end{equation}
The constant $C_{\mathrm{rate}}$ does not depend on the initial distribution $\mu$.
\end{proposition}

\subsection{A Regularity Property of \texorpdfstring{$\gamma$}{gamma}}
\label{ssec:usc}

Since $\gamma$ governs the separation condition for clustering (Assumption~\ref{ass:separation} below), we record how it behaves under perturbations of the transition matrices. Let $\mathcal S^\circ_d$ denote the set of transition matrices of irreducible aperiodic Markov chains on $\Sigma$, endowed with the entrywise norm $\|P-P'\|_\infty=\max_{a,b}|P(a,b)-P'(a,b)|$.

\begin{proposition}[Upper semi-continuity of $\gamma$]
\label{prop:usc}
The map $\gamma:\mathcal S^\circ_d\times\mathcal S^\circ_d\to\R^+$ is upper semi-continuous: if $(P_m,Q_m)\to(P_0,Q_0)$, then
\[
\limsup_{m\to\infty}\gamma(P_m,Q_m) \;\le\; \gamma(P_0,Q_0).
\]
\end{proposition}

\begin{remark}
\label{rmk:usc-consequence}
For every $c>0$, the set $\{P\in\mathcal S^\circ_d:\gamma(P,P)<c\}$ is open: low within-cluster dispersion is stable under perturbation of the transition matrix.
\end{remark}

\subsection{Bounds on \texorpdfstring{$\gamma$}{gamma}}
\label{sec:bounds}

\begin{proposition}[Wasserstein lower bound; product upper bound]
\label{prop:bounds}
Let $P,Q$ be irreducible transition matrices on $\Sigma$ with stationary distributions $\pi_P,\pi_Q$, let $S:=(c_{\sub}(a,b))_{a,b\in\Sigma}$, and let $\Wd$ denote the $1$-Wasserstein distance on $(\Sigma_\bot,\bar d)$. Then
\begin{equation}
\label{eq:bounds}
\Wd(\pi_P,\pi_Q) \;\le\; \gamma(P,Q) \;\le\; \pi_P^\top S\,\pi_Q .
\end{equation}
\end{proposition}

Both bounds depend on the considered Markov chains through their stationary distributions alone, and are computable from them: the upper bound in closed form, the lower one by linear programming. The lower bound is informative in that $\Wd(\pi_P,\pi_Q)>0$ as soon as $\pi_P\neq\pi_Q$, so $\gamma(P,Q)>0$ whenever the two chains have different stationary distributions.

\begin{remark}[Closed forms]
Under the default scheme of TraMineR, $c_{\sub}\equiv2$ off the diagonal and $\delta\equiv1$, and $\dom(x_{1:n},y_{1:n})=2\bigl(n-\mathrm{LCS}(x_{1:n},y_{1:n})\bigr)$, so $\gamma=2(1-\lambda)$ with $\lambda$ the limiting normalized LCS score of the two chains. Computing $\gamma$ in closed form therefore generalizes the Chv\'atal--Sankoff problem~\cite{chvatal1975longest}, which remains open for i.i.d.\ uniform binary sequences despite tight numerical bounds, in two directions at once: Markov dependence within each sequence and an arbitrary substitution cost in place of the one that makes the LCS appear.
\end{remark}

\begin{remark}
The subadditivity~\eqref{eq:subadd} and the normalization by $n$ both require sequences of equal length. The appropriate normalization for sequences of unequal lengths, and the population quantity it estimates, remain to be identified. Extensions of OM to continuous-time trajectories have been proposed \cite{rama2019aliclu}; establishing an analogous asymptotic theory in that setting is also an open problem.
\end{remark}

\section{Clustering Consistency}
\label{sec:clustering}

\subsection{The Mixture Model}

Let $K\ge 2$ time-homogeneous irreducible aperiodic Markov chains on $\Sigma$, with transition matrices $P_1,\ldots,P_K$, arbitrary initial distributions $\mu_1,\ldots,\mu_K$, and mixture weights $w_1,\ldots,w_K\in(0,1)$ such that $\sum_k w_k=1$. We observe $N$ sequences $X_1,\ldots,X_N$ of length $n$, generated as follows: latent labels $Z_1,\ldots,Z_N$ are drawn i.i.d.\ with $\PP(Z_i=k)=w_k$. The $X_i$ are length-$n$ independent realizations of the Markov chain with transition matrix $P_{Z_i}$ and initial distribution $\mu_{Z_i}$. 
Let
\[
G_k:=\{i:Z_i=k\},
\qquad
w_{\min}:=\min_k w_k,
\qquad
w_{\max}:=\max_k w_k,
\]
and set
\[
\Gamma_{kk'}:=\gamma(P_k,P_{k'}),
\qquad
\hat\gamma_n(i,j):=\frac1n\dom(X_i,X_j).
\]
By Proposition~\ref{prop:convergence}, $\Gamma_{kk'}$ is well defined and does not depend on the initial distributions. Let $\mathcal P^\star$ denote the partition of $\{1,\ldots,N\}$ into the non-empty classes $\{G_k:G_k\neq\emptyset\}$.

Throughout this section, the number of clusters $K$, the component transition matrices, initial distributions, mixture weights, and the OM cost scheme are fixed, while $n\to\infty$ and $N=N_n\to\infty$.

\begin{assumption}[Growth condition]
\label{ass:growth}
$N=N_n$ satisfies $N_n\xrightarrow[n\to\infty]{}\infty$ and $\log N_n\underset{n\to\infty}{=}o(n)$.
\end{assumption}

\begin{assumption}[Separation condition]
\label{ass:separation}
Let
\[
\gin:=\max_k\Gamma_{kk},
\qquad
\gout:=\min_{k\neq k'}\Gamma_{kk'}.
\]
We assume
\[
\eta:=\gout-\gin>0.
\]
\end{assumption}

Assumption~\ref{ass:separation} can be interpreted as follows: the maximum within-cluster asymptotic dispersion $\gin$ must be strictly smaller than the minimum between-cluster asymptotic separation $\gout$.

For $k,k'\in\{1,\ldots,K\}$, define the finite-horizon mean
\begin{equation}
\label{eq:Gamma-n}
\Gamma^{(n)}_{kk'}
:=
\E\!\left[
\hat\gamma_n(i,j)
\,\middle|\,
Z_i=k,Z_j=k'
\right],
\qquad i\neq j.
\end{equation}

The first key intermediate result is the uniform concentration of the empirical distance matrix around its finite-horizon mean. It extends the pairwise concentration bound to all distances used by the clustering algorithms simultaneously.

\begin{lemma}[Uniform concentration of the distance matrix]
\label{lem:uniform}
Fix $\varepsilon>0$ and define
\[
\mathcal E_{N,n}(\varepsilon)
:=
\left\{
\forall\,1\le i<j\le N:
\left|
\hat\gamma_n(i,j)-\Gamma^{(n)}_{Z_iZ_j}
\right|
\le\varepsilon
\right\}.
\]
Set
\[
C^\star
:=
\max_{1\le k,k'\le K} C(M,P_k,P_{k'}),
\]
with $C(\cdot)$ defined in~\eqref{eq:Cdef}. Under
Assumption~\ref{ass:metric}, for every $n\ge1$,
\begin{equation}
\label{eq:uniform-centered-bound}
\PP\bigl(\mathcal E_{N,n}(\varepsilon)^c\bigr)
\le
N^2
\exp\!\left(
-\frac{2\varepsilon^2n}{C^\star}
\right).
\end{equation}
In particular, under Assumption~\ref{ass:growth},
\[
\PP\bigl(\mathcal E_{N,n}(\varepsilon)\bigr)\longrightarrow1
\qquad\text{as }n\to\infty.
\]
\end{lemma}

\subsection{Hierarchical Clustering}

\begin{definition}[Hierarchical agglomerative clustering]
\label{def:HAC}
Let $L$ be a linkage, i.e.~a function assigning a non-negative number to every pair of disjoint subsets of $\{1,\ldots,N\}$. Hierarchical agglomerative clustering constructs a sequence of partitions $(\mathcal P_\ell)_{\ell=0}^{N-1}$ starting from $\mathcal P_0:=(\{1\},\ldots,\{N\})$ and, at step $\ell$, merging a pair $(A_\ell,B_\ell)\in\arg\min_{C\neq D\in\mathcal P_\ell}L(C,D)$ to form
\[
\mathcal P_{\ell+1}
=
(\mathcal P_\ell\setminus\{A_\ell,B_\ell\})
\cup
\{A_\ell\cup B_\ell\}.
\]
The procedure terminates at $\mathcal P_{N-1}=\{\{1,\ldots,N\}\}$.
\end{definition}

For \emph{single linkage}, $L(A,B):=\min_{i\in A,j\in B}\hat\gamma_n(i,j).$ Equivalently, for every $t\ge0$, define the graph $\mathcal G_t:=(\{1,\ldots,N\},E_t)$ with $(i,j)\in E_t$ if and only if $\hat\gamma_n(i,j)\le t$. The partition obtained by single linkage at
level $t$ is the set of connected components of $\mathcal G_t$.

When the number of clusters $K$ is known, let
\[
\widehat{\mathcal P}^{\mathrm{SL}}_{N,n}(K):=\mathcal P_{N-K}
\]
denote the partition obtained by cutting the single-linkage dendrogram
when $K$ blocks remain.

\begin{theorem}[Consistency of single-linkage clustering]
\label{thm:SL}
Under Assumptions~\ref{ass:metric} and~\ref{ass:separation}, let
$\varepsilon\in(0,\eta/2)$. There exists a constant $B \in \R_+$,
depending only on the component laws and the OM cost scheme, such
that, for every $n\ge2$ and $N\ge K$ satisfying
\[
B\sqrt{\frac{\log n}{n}}+\varepsilon<\frac{\eta}{2},
\]
\[
\PP\!\left(
\widehat{\mathcal P}^{\mathrm{SL}}_{N,n}(K)
\neq
\mathcal P^\star
\right)
\le
N^2\exp\!\left(
-\frac{2\varepsilon^2n}{C^\star}
\right)
+
K(1-w_{\min})^N.
\]
In particular, under Assumption~\ref{ass:growth},
\[
\PP\!\left(
\widehat{\mathcal P}^{\mathrm{SL}}_{N,n}(K)
=
\mathcal P^\star
\right)
\longrightarrow1
\qquad\text{as }n\to\infty.
\]
\end{theorem}

\begin{remark}[Other linkages]
\label{rmk:other-linkages}
The same conclusion holds for any bracketed linkage, that is any linkage satisfying $\min_{i\in A,j\in B}\hat\gamma_n(i,j) \le L(A,B) \le \max_{i\in A,j\in B}\hat\gamma_n(i,j)$ for all disjoint $A,B$. Hence, a similar result holds for complete and average linkage. Ward, centroid and median linkages are not bracketed and are not covered. 
\end{remark}


\subsection{\texorpdfstring{$K$}{K}-Medoids (PAM)}

\begin{definition}[$K$-medoids]
\label{def:kmedoids}
Let $\mathcal M\subseteq{\{1,\ldots,N\}}$, with $|\mathcal M|=K$. The $K$-medoids objective is

$$
\Phi_{N,n}(\mathcal M)
:=
\frac1N
\sum_{i=1}^N
\min_{m\in\mathcal M}
\hat\gamma_n(i,m).
$$

The $K$-medoids problem is to minimize $\Phi_{N,n}(\mathcal M)$. The minimization ranges over all subsets $\mathcal M\subseteq\{1,\ldots,N\}$ with $|\mathcal M|=K$.
\end{definition}

This problem is NP-hard in general \cite{kariv1979}. In practice, it is commonly addressed by Partitioning Around Medoids (PAM), using a greedy initialization followed by local improvement through medoid swaps \cite{schubert2021}. 

\medskip
\noindent
\textbf{PAM algorithm.}
Starting from an initial set $\mathcal M^{(0)}$ of $K$ medoids, PAM repeatedly replaces one current medoid by one non-medoid whenever this strictly decreases $\Phi_{N,n}$, and stops when no improving swap exists. We denote the final set of medoids by $\widehat{\mathcal M}^{\mathrm{PAM}}_{N,n}$ and the corresponding
nearest-medoid partition by $\widehat{\mathcal P}^{\mathrm{PAM}}_{N,n}$.

In the sequel, fix
\begin{equation}
\label{eq:r-pam}
0<r<\frac{w_{\min}\eta}{16},\qquad \delta_0:=\frac{w_{\min}}{2},\qquad N_0:=\left\lceil\frac{16(\gin+r)}{w_{\min}\eta}\right\rceil.
\end{equation}
Let $B$ be the constant appearing in Theorem~\ref{thm:SL}, and set
\[
b_n:=B\sqrt{\frac{\log n}{n}}.
\]

We define the \emph{label-frequency event} by
\begin{equation}
\label{eq:Wevent}
\mathcal W_N(\delta_0):=\left\{\forall \ell\in\{1,\ldots,K\}: |\hat w_\ell-w_\ell|\le\delta_0\right\},\qquad \hat w_\ell:=\frac{|G_\ell|}{N} = \frac{1}{N}\sum_{i=1}^N \ind_{\{Z_i = \ell\}}.
\end{equation}
Here $G_\ell$ denotes the set of points in the $\ell$-th cluster. Since the $Z_i$ are i.i.d., Hoeffding's inequality and a union bound give
\begin{equation}
\label{eq:Wbound}
\PP\bigl(\mathcal W_N(\delta_0)^c\bigr)\le 2K\exp\bigl(-2N\delta_0^2\bigr).
\end{equation}
On $\mathcal W_N(\delta_0)$ one has $\hat w_\ell\ge w_{\min}/2>0$ for every $\ell$; in particular no class is empty.

\begin{lemma}[Every cluster has exactly one medoid]
\label{lem:kmedoid-one-per-cluster}
Under Assumptions~\ref{ass:metric} and~\ref{ass:separation}, suppose that $b_n<r$ and $N\ge\max{(N_0,K)}$. On the event $\mathcal E_{N,n}(r-b_n)\cap\mathcal W_N(\delta_0)$, the PAM output $\widehat{\mathcal M}^{\mathrm{PAM}}_{N,n}$ contains exactly one index from each cluster $G_1,\ldots,G_K$.
\end{lemma}

\begin{theorem}[Consistency of PAM]
\label{thm:kmedoids}
Under Assumptions~\ref{ass:metric} and~\ref{ass:separation}, with $r,\delta_0,N_0$ as in~\eqref{eq:r-pam}, for every $n\ge2$ and $N\ge\max{(N_0,K)}$ such that

$$
B\sqrt{\frac{\log n}{n}}<r,
$$

$$
\PP\bigl(\widehat{\mathcal P}^{\mathrm{PAM}}_{N,n}\neq\mathcal P^\star\bigr)
\le
N^2\exp\!\left(-\frac{2\left(r-B\sqrt{\frac{\log n}{n}}\right)^2n}{C^\star}\right)
+
2K\exp\bigl(-2N\delta_0^2\bigr).
$$

In particular, under Assumption~\ref{ass:growth},

$$
\PP\bigl(\widehat{\mathcal P}^{\mathrm{PAM}}_{N,n}=\mathcal P^\star\bigr)\xrightarrow[n\to\infty]{}1.
$$

\end{theorem}


\subsection{Selecting the Number of Clusters}
\label{ssec:K-selection}

Theorems~\ref{thm:SL} and~\ref{thm:kmedoids} assume that $K$ is known. We now give a consistent estimator of $K$ based on the empirical OM distance matrix.

For $N\ge3$ and $i\neq j$, let
\begin{equation}
\label{eq:profile-distance}
\rho_{N,n}(i,j) := \max_{\ell\notin\{i,j\}} \left| \hat\gamma_n(i,\ell)-\hat\gamma_n(j,\ell)\right|.
\end{equation}
Let $h_1^\rho\le\cdots\le h_{N-1}^\rho$ be the single-linkage merge heights associated with $\rho_{N,n}$, and set
\[
h_{\mathrm{med}}^\rho:=h_{\lceil (N-1)/2\rceil}^\rho,
\qquad
h_{\max}^\rho:=h_{N-1}^\rho.
\]
Define the data-driven threshold
\begin{equation}
\label{eq:profile-threshold}
a_{N,n}
:=
\max\left\{
\sqrt{h_{\mathrm{med}}^\rho h_{\max}^\rho},
\;
h_{\max}^\rho\left(\frac{\log N}{n}\right)^{1/4}
\right\},
\end{equation}
and let $H_{N,n}$ be the graph on $\{1,\ldots,N\}$ in which $i$ and $j$ are adjacent whenever $\rho_{N,n}(i,j)\le a_{N,n}$.

We define the following estimator of the number of clusters:
\begin{equation}
\label{eq:Khat-profile}
\hat K_{N,n} := \#\{\text{connected components of }H_{N,n}\}.
\end{equation}

\begin{theorem}[Consistency of the number of clusters]
\label{thm:K-selection}
Under Assumptions~\ref{ass:metric} and~\ref{ass:separation}, let $N \ge 3$ and let $B \in \R_+$ be the constant appearing in Theorem~\ref{thm:SL}. Set
\[
x_{N,n}:=\frac{\log N}{n},
\qquad
\varepsilon_{N,n}:=x_{N,n}^{3/8},
\qquad
\beta_{N,n}:=\eta-2B\sqrt{\frac{\log n}{n}}-2\varepsilon_{N,n}.
\]
Whenever
\[
2\varepsilon_{N,n}<\beta_{N,n}x_{N,n}^{1/4},
\qquad
\sqrt{2M\varepsilon_{N,n}}<\beta_{N,n},
\qquad
Mx_{N,n}^{1/4}<\beta_{N,n},
\]
we have
\[
\PP\bigl(\hat K_{N,n}\neq K\bigr)
\le
N^2\exp\!\left(
-\frac{2\varepsilon_{N,n}^2n}{C^\star}
\right)
+
K\left[
(1-w_{\min})^N
+
N(1-w_{\min})^{N-1}
\right].
\]
In particular, under Assumption~\ref{ass:growth},
\[
\PP\bigl(\hat K_{N,n}=K\bigr)\longrightarrow1
\qquad\text{as }n\to\infty.
\]
\end{theorem}

\begin{remark}[Clustering with unknown $K$]
\label{rmk:unknown-K}
Theorem~\ref{thm:K-selection} can be combined with either Theorem~\ref{thm:SL} or Theorem~\ref{thm:kmedoids} to obtain a procedure that does not require prior knowledge of $K$. First compute $\hat K_{N,n}$ from $H_{N,n}$, and then run hierarchical clustering or PAM with $\hat K_{N,n}$ clusters. On the event $\{\hat K_{N,n}=K\}$, the second step coincides with the corresponding known-$K$ procedure. Hence, under Assumptions~\ref{ass:metric}, \ref{ass:growth}, and \ref{ass:separation}, the resulting clustering is consistent.
\end{remark}

\begin{remark}
Lee et al.~\cite{nearoptimal} derive a lower bound on the clustering error for mixtures of Markov chains. In particular, their bound implies that asymptotically exact recovery requires $n\mathcal D=\Omega(\log N)$, where $\mathcal D$ is a weighted Kullback--Leibler divergence between the component kernels. For fixed component laws satisfying Assumption~\ref{ass:separation}, our recovery bounds show that $n$ of order $\log N$ is sufficient. Thus, our logarithmic dependence on $N$ matches theirs, although the separation quantities $\mathcal D$ and $\eta$ are different.
\end{remark}

\section{An Extension to Multichannel Hidden Markov Models}
\label{sec:extensions}

The clustering results of Section~\ref{sec:clustering} rely on three ingredients: (i) the concentration results of the pairwise distances  (Propositions~\ref{prop:concentration} and~\ref{prop:diagonal-rate}), (ii) the uniform exponential concentration of $\hat\gamma_n(i,j)$ around $\Gamma^{(n)}_{Z_iZ_j}$ and (iii) deterministic combinatorial arguments conditional on $\mathcal E_{N,n}(\varepsilon)$ and $\mathcal W_N(\delta_0)$. Ingredient~(i) has direct hidden-Markov analogues: concentration follows from the bounded-difference inequality for hidden-Markov observations of \cite[Example~2.15 and Corollary~2.16]{paulin}, while the diagonal rate follows by adapting \cite[Corollary~3.1]{houdre2019} to the OM-induced alignment score through Lemma~\ref{lem:OM-score-properties}. Ingredients~(ii) and (iii) use no further property of the observed sequences. We now phrase these observations as extensions of the clustering results to multichannel hidden Markov models.

We consider a mixture of Hidden Markov Models suited to multichannel categorical sequences. Let $K$ be the number of components of the mixture. Let $\mathcal Z$ be a finite set, called hidden state space, and a finite, possibly multichannel alphabet $\Sigma=\Sigma_1\times\cdots\times\Sigma_J$ with $J\ge1$. For each $k \in \{1, \ldots, K\}$, the $k$-th component of the mixture model is specified by
\begin{enumerate}[label=(\roman*)]
\item  An irreducible aperiodic Markov chain $(Z^{(k)}_t)_{t\ge1}$ on $\mathcal Z$, with transition matrix $R_k$ and arbitrary initial law, called the hidden chain;

\item A measurable emission map $\varphi_k:\mathcal Z\times\mathcal U\to\Sigma$, where $(\mathcal U,\nu)$ is a Polish space endowed with a probability measure $\nu$. 
\end{enumerate}

The observed process associated with the $k$-th component of the model is $X^{(k)}_t:=\varphi_k(Z^{(k)}_t,\xi_t)$ with $(\xi_t)\stackrel{\mathrm{iid}}{\sim}\nu$ independent of the hidden chain.

We use the multichannel OM distance as defined in \cite{multichannelSA}. It is computed on $\Sigma^\star$ with the multichannel substitution cost $c^{\mathrm{mc}}_{\sub}(a,b):=\sum_{j=1}^J\lambda_j\,c^{(j)}_{\sub}(a_j,b_j)$ and multichannel gap cost $\delta^{\mathrm{mc}}(a):=\sum_{j=1}^J\lambda_j\,\delta^{(j)}(a_j)$, where $\lambda_1,\ldots,\lambda_J>0$ are channel weights. Set $M^{\mathrm{mc}}:=\max_{a,b}c^{\mathrm{mc}}_{\sub}(a,b)$.

Assumption~\ref{ass:metric} extends to the multichannel case: if each $(c^{(j)}_{\sub},\delta^{(j)})$ satisfies (i)--(vi) on $\Sigma_j$, then so does $(c^{\mathrm{mc}}_{\sub},\delta^{\mathrm{mc}})$ on $\Sigma$.

\begin{proposition}[Convergence, concentration and diagonal rate for multichannel HMMs]
\label{prop:gen-concentration}
Assume that $(c^{\mathrm{mc}}_{\sub},\delta^{\mathrm{mc}})$ on $\Sigma$ satisfies Assumption~\ref{ass:metric}, and let $X^{(k)}$ and $X^{(k')}$ be independent observed processes of components $k$ and $k'$ under the above model. Set
\[
\hat\gamma_n^{\mathrm{mc}}
:=
\frac1n\dom(X^{(k)}_{1:n},X^{(k')}_{1:n}).
\]
Then,
\begin{enumerate}[label=(\roman*),leftmargin=*,topsep=2pt]
\item $\hat\gamma_n^{\mathrm{mc}}$ converges almost surely to a deterministic limit $\gamma^{\mathrm{mc}}(k,k')$, which depends only on the laws of the observed processes and satisfies
\[
\gamma^{\mathrm{mc}}(k,k')
=
\inf_{n\ge1}
\E_{\pi_k,\pi_{k'}}[\hat\gamma_n^{\mathrm{mc}}],
\]
where $\pi_k,\pi_{k'}$ denote the stationary laws of the hidden chains;
\item setting $C^{\mathrm{mc}}_{kk'} := C_{\mathrm{Pau}}(M^{\mathrm{mc}})^2\tmix(R_k\otimes R_{k'})$, for every $n\ge1$ and every $s>0$,
\[
\PP\left(
\left|
\hat\gamma_n^{\mathrm{mc}}
-
\E[\hat\gamma_n^{\mathrm{mc}}]
\right|
\ge s
\right)
\le
2\exp\left(
-\frac{2s^2n}{C^{\mathrm{mc}}_{kk'}}
\right);
\]
\item if $k=k'$, there exists a constant $C^{\mathrm{mc}}_{\mathrm{rate},k} \in \R_+$ such that, for every $n\ge2$,
\[
\left|
\E[\hat\gamma_n^{\mathrm{mc}}]
-
\gamma^{\mathrm{mc}}(k,k)
\right|
\le
C^{\mathrm{mc}}_{\mathrm{rate},k}
\sqrt{\frac{\log n}{n}}.
\]
\end{enumerate}
\end{proposition}

Consider now $N$ independent sequences of length $n$ drawn from the mixture of the $K$ components with weights $w_1,\ldots,w_K$, and set $\Gamma^{\mathrm{mc}}_{kk'}:=\gamma^{\mathrm{mc}}(k,k')$. The separation condition becomes the following.

\begin{assumption}[Multichannel separation condition]
\label{ass:mc-separation}
$\eta^{\mathrm{mc}}:=\min_{k\neq k'}\Gamma^{\mathrm{mc}}_{kk'}-\max_k\Gamma^{\mathrm{mc}}_{kk}>0$.
\end{assumption}

\begin{theorem}[Clustering consistency for multichannel HMMs]
\label{thm:gen-clustering}
Under the above model, Assumption~\ref{ass:metric} applied to $(c^{\mathrm{mc}}_{\sub},\delta^{\mathrm{mc}})$, and Assumptions~\ref{ass:growth} and~\ref{ass:mc-separation}, let $\widehat{\mathcal P}^{\mathrm{SL}}_{N,n}(K)$, $\widehat{\mathcal P}^{\mathrm{PAM}}_{N,n}$ and $\hat K_{N,n}$ denote the analogues of the estimators defined in Section~\ref{sec:clustering}, computed from the multichannel OM distances. Then,
\begin{enumerate}[label=(\roman*),leftmargin=*,topsep=2pt]
\item (Hierarchical clustering) $\PP(\widehat{\mathcal P}^{\mathrm{SL}}_{N,n}(K)=\mathcal P^\star)\to1$, and the same holds for any bracketed linkage;
\item (PAM) $\PP(\widehat{\mathcal P}^{\mathrm{PAM}}_{N,n}=\mathcal P^\star)\to1$;
\item (Number of clusters) $\PP(\hat K_{N,n}=K)\to1$.
\end{enumerate}
\end{theorem}

\begin{remark}
Distinct specifications may induce the same law on the observed process \cite{gilbert1959}; then $\gamma^{\mathrm{mc}}(k,k')=\gamma^{\mathrm{mc}}(k,k)$ and the separation condition fails. Indeed, the OM distance operates on observed sequences and cannot distinguish observationally equivalent models. The separation condition is therefore a condition on the laws of the observed process, which is the natural clustering criterion.
\end{remark}

\section{Numerical Illustrations}
\label{sec:numerics}

The OM distance is computed by a dedicated Numba implementation of the dynamic program, checked to return exactly the values of TraMineR's \texttt{seqdist} from version 2.2.12 on the cost schemes considered here. Scripts, notebooks and figures are available at \url{https://github.com/OttaKhalifa/OMConvergence}.

\subsection{Convergence of \texorpdfstring{$\hat\gamma_n$}{gamma-hat n}}
\label{ssec:num-convergence}

We first illustrate Proposition~\ref{prop:convergence} and the bounds of Proposition~\ref{prop:bounds} on a single pair of chains.

\paragraph{Design.} We take $\Sigma=\{0,\ldots,4\}$ and draw the rows of two transition matrices $P,Q$ i.i.d.\ from a Dirichlet distribution with concentration parameter $\alpha=0.25$. We consider two configurations, each with $R=30$ independent replicates: a \emph{within} configuration, where $X$ and $Y$ are independent with the same kernel $P$, so that $\hat\gamma_n$ estimates $\gamma(P,P)$; and a \emph{between} configuration, where $X\sim P$ and $Y\sim Q$, estimating $\gamma(P,Q)$.

To illustrate the independence of the limit from the initial law granted by Proposition~\ref{prop:convergence}, each of the $120$ simulated trajectories starts from the Dirac mass at a state drawn uniformly on $\Sigma$, independently across trajectories. We evaluate $\hat\gamma_n$ at $24$ horizons on a logarithmic scale between $n=10$ and $n=10^4$, along increasing prefixes of the same two trajectories.

In Figure~\ref{fig:gamma-convergence}, we use the classical sequence-analysis default substitution cost scheme, $c_{\sub}\equiv2$ off the diagonal. In the appendix, we also test a transition-rate (TRATE) substitution cost scheme, $c_{\sub}=2-\hat T-\hat T^\top$ with $\hat T$ the empirical transition-rate matrix of a sample of two trajectories of length $5000$, one from each kernel, independent of the sequences entering $\hat\gamma_n$, and a random symmetric scheme with $c_{\sub}\sim\mathcal U[1.2,2]$. All three take $\delta\equiv1$ for both insertions and deletions. The three schemes are applied to the same $120$ trajectories.

\begin{figure}[tbp]
\centering
\includegraphics[width=0.7\linewidth]{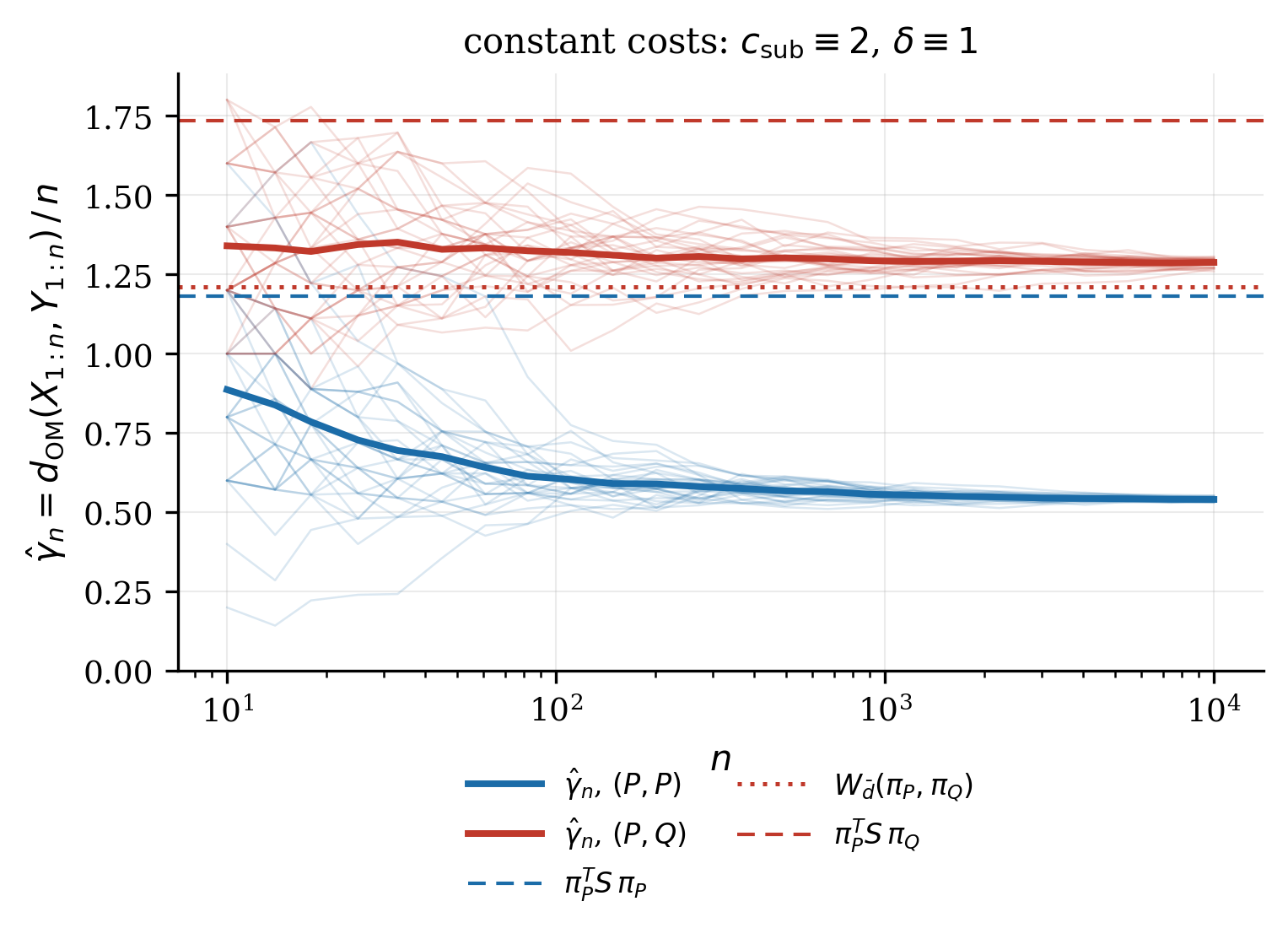}
\caption{Sample paths of $\hat\gamma_n$ for the within (blue) and between (red) configurations, with the bounds of Proposition~\ref{prop:bounds} as horizontal references. Constant cost scheme, $30$ replicates.}
\label{fig:gamma-convergence}
\end{figure}

The paths stabilize well before $n=10^4$ and the spread becomes negligible. The within-pair limit is bounded away from zero. Both bounds from Proposition~\ref{prop:bounds} hold, the Wasserstein one being the tighter.

\subsection{Plausibility of the assumption on \texorpdfstring{$\gamma$}{gamma}}
\label{ssec:num-separation}

We now assess how often Assumption~\ref{ass:separation} holds under a random mixture. The condition bears on the limits $\Gamma_{k\ell}=\gamma(P_k,P_\ell)$ of Proposition~\ref{prop:convergence}, and reads $\eta>0$ for the signed margin, in OM cost units,
\[
\eta=\min_{k\neq\ell}\Gamma_{k\ell}-\max_{k}\Gamma_{kk}.
\]

\paragraph{Design.} We draw $K$ transition matrices on $\Sigma=\{0,\ldots,4\}$ with rows i.i.d.\ Dirichlet($\alpha$) and uniform weights $1/K$, small $\alpha$ giving easily distinguishable chains, large $\alpha$ near-uniform ones. Each $\Gamma^{(n)}_{k\ell}$ is estimated on $60$ independent trajectory pairs, each trajectory started at a state drawn uniformly on $\Sigma$. Within each repetition, we construct simultaneous $95\%$ Student confidence intervals for the $K(K+1)/2$ distinct entries of $\Gamma^{(n)}$, using a Bonferroni correction over these entries, and deduce confidence intervals for $\eta_n$: a repetition is declared separated when the resulting interval lies above zero, nonseparated when it lies below, and undecided otherwise. We take $\alpha\in\{0.2,0.3,0.4,0.5,1,5,10\}$, $K\in\{2,\ldots,10\}$, $n=1000$ and $R=30$ repetitions per $(\alpha, K)$ cell, with new chains drawn at every repetition. Monte Carlo proportions across these $R=30$ repetitions are quantified using $95\%$ Wilson score confidence intervals. The cost scheme is $c_{\sub}\equiv2$, $\delta\equiv1$, fixed in advance.

Assumption~\ref{ass:separation} is required in Theorems~\ref{thm:SL} and~\ref{thm:kmedoids}. The heatmap in Figure~\ref{fig:assumption-separation} delimits the region over which the consistency assumptions of both hierarchical agglomerative clustering and PAM are plausible.

\begin{figure}[tbp]
\centering
\includegraphics[width=0.7\linewidth]{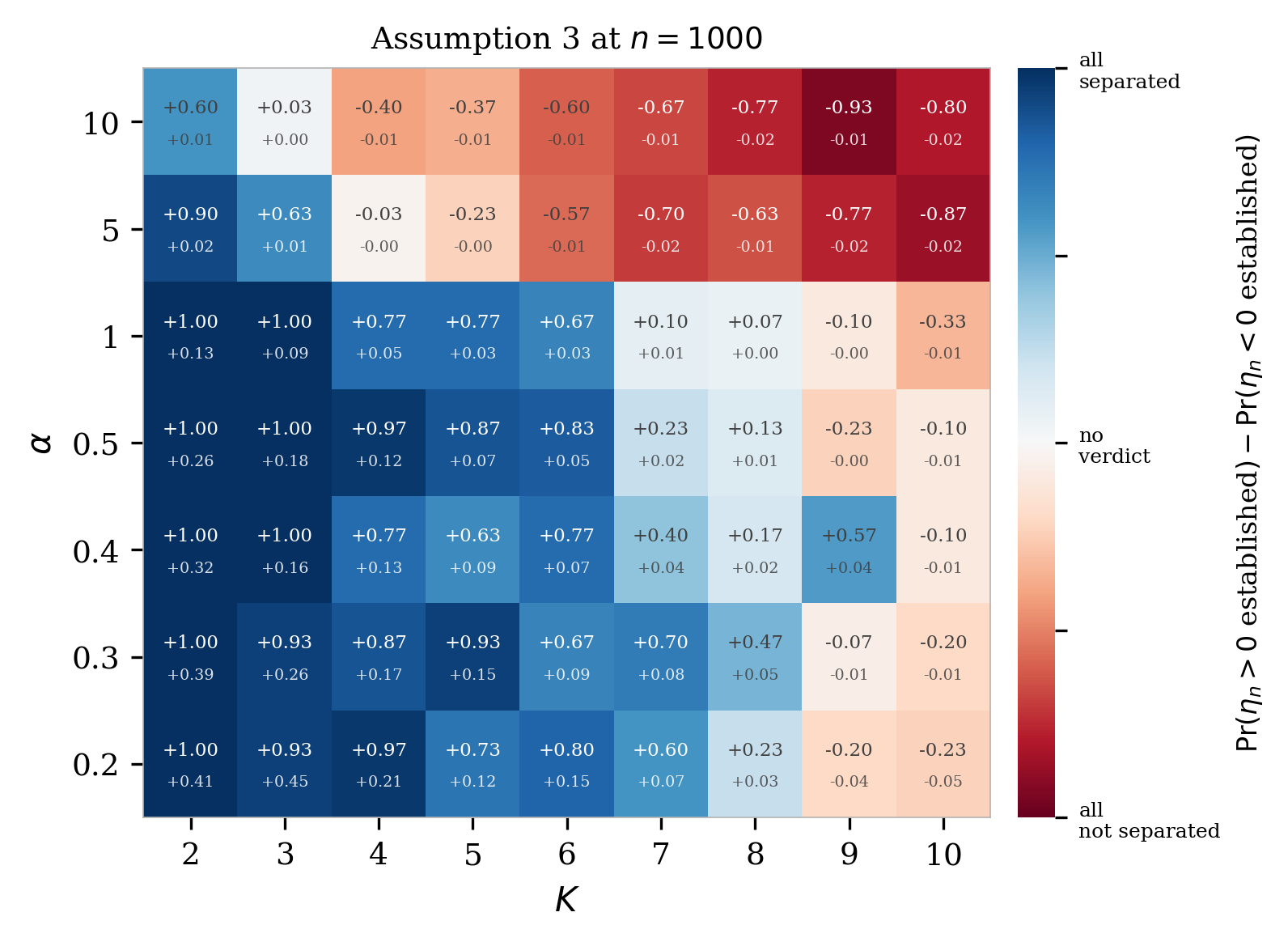}
\caption{Heatmap of the verdict balance $\PP(\text{separated})-\PP(\text{nonseparated})$ at level $0.95$, as a function of $\alpha$ (rows) and $K$ (columns), with signed median $\hat\eta_n$ (below) in each cell.}
\label{fig:assumption-separation}
\end{figure}

Finite-horizon separation is established over a substantial part of the grid: with probability at least $0.8$ for $K\le4$ throughout $\alpha\le1$, and in every repetition at $K=2$. It falls to a third or less at $K=10$, and all but vanishes beyond $K=5$ once $\alpha\ge5$. The decay in $K$ is structural since $\gin$ is a maximum over $K$ terms and $\gout$ a minimum over $K(K-1)/2$, so the two are driven towards each other as $K$ grows.

\subsection{Recovery of the true partition}
\label{ssec:num-recovery}

We now measure how hierarchical clustering with average linkage (more frequently used in practice than single linkage for stability reasons) and $K$-medoids behave at finite sample size, first along a path of increasing $n$ with fixed $N$, then over the same $(\alpha,K)$ grid as Section~\ref{ssec:num-separation}.

\paragraph{Design.} Mixtures are drawn as in Section~\ref{ssec:num-separation} with $\alpha=1$ and $K=4$. The cost scheme is $c_{\sub}\equiv2$, $\delta\equiv1$. The distance matrix $(\hat\gamma_n(i,j))_{i,j\le N}$ is computed and clustered by average linkage, using the SciPy version 1.15.3 implementation, and $K$-medoids with the one-swap PAM algorithm, with known number of clusters $K=4$. We report the Adjusted Rand Index against the true partition $\mathcal{P}^\star$. We draw $N = 800$ sequences and take lengths $n$ ranging from $50$ to $500$, along increasing prefixes. Figure~\ref{fig:ARIconvergence} reports the resulting ARI trajectories.

\begin{figure}[tbp]
\centering
\begin{minipage}[t]{0.49\linewidth}
  \centering
  \includegraphics[width=\linewidth]{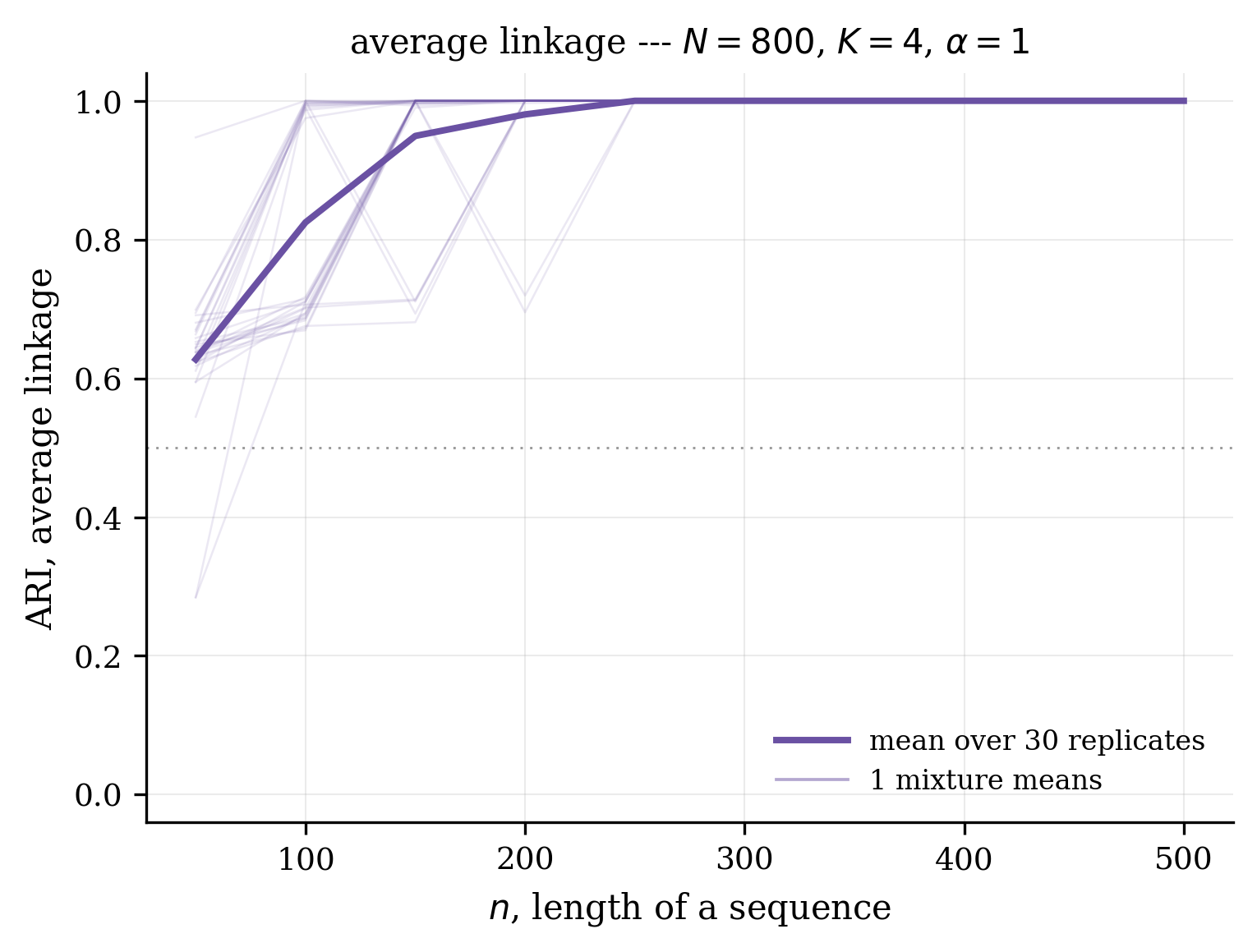}

  (a) Average linkage.
\end{minipage}
\hfill
\begin{minipage}[t]{0.49\linewidth}
  \centering
  \includegraphics[width=\linewidth]{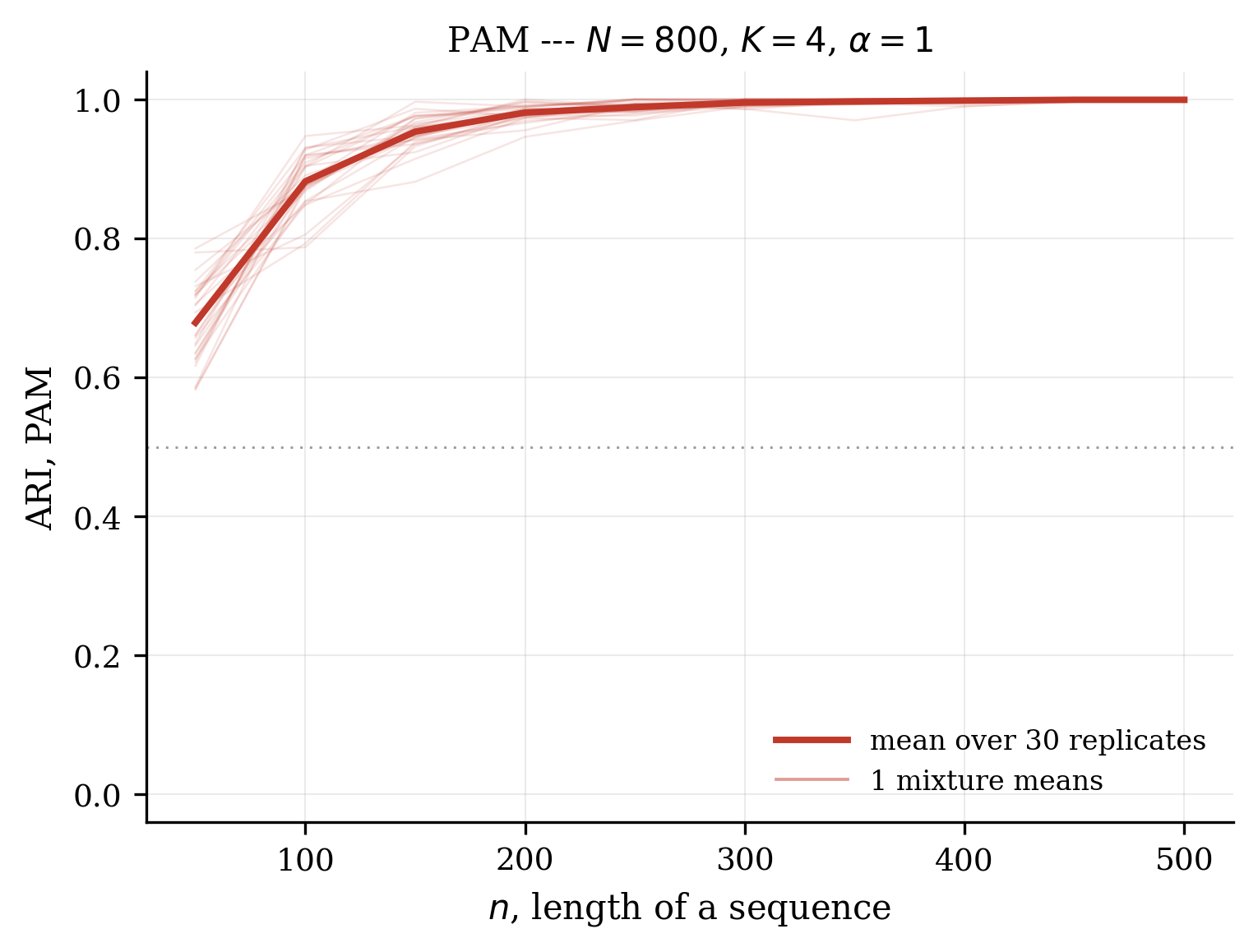}

  (b) K-medoids (PAM).
\end{minipage}
\caption{$\mathrm{ARI}$ trajectories in function of $n$ for fixed $N=800,$ $K=4$ and $\alpha=1$, with $R=50$ replications. }
\label{fig:ARIconvergence}
\end{figure}

For average linkage, the mean ARI rises from $0.6$ to $1$ over the range of $n$ covered by the path, reaching $1$ at $n=300$. Some ARI trajectories are not monotone, exact recovery being lost and regained as $n$ grows: this is due to the well-known instability of such agglomerative approaches. The same holds for single-linkage, whose ARI trajectory is provided in the Appendix and achieves full recovery slower in $n$. $K$-medoids is slightly more accurate at small $n$ and more stable with a mean ARI ranging from $0.7$ to $1$ and reaching $1$ at $n=300$, with fewer non-monotone trajectories and smaller ARI drops.

\paragraph{The $(\alpha,K)$ grid.} At $N=800$ and $n=1000$ we run the same procedure over $\alpha \in \{0.2,0.3,0.4,0.5,1,5,10\}$ and $K\in\{2,\ldots,10\}$, with $R=30$ repetitions per cell and new kernels drawn at every repetition. The grid matches that of Section~\ref{ssec:num-separation} cell for cell and is reported in Figure~\ref{fig:recovery}.

\begin{figure}[tbp]
\centering
\includegraphics[width=1\linewidth]{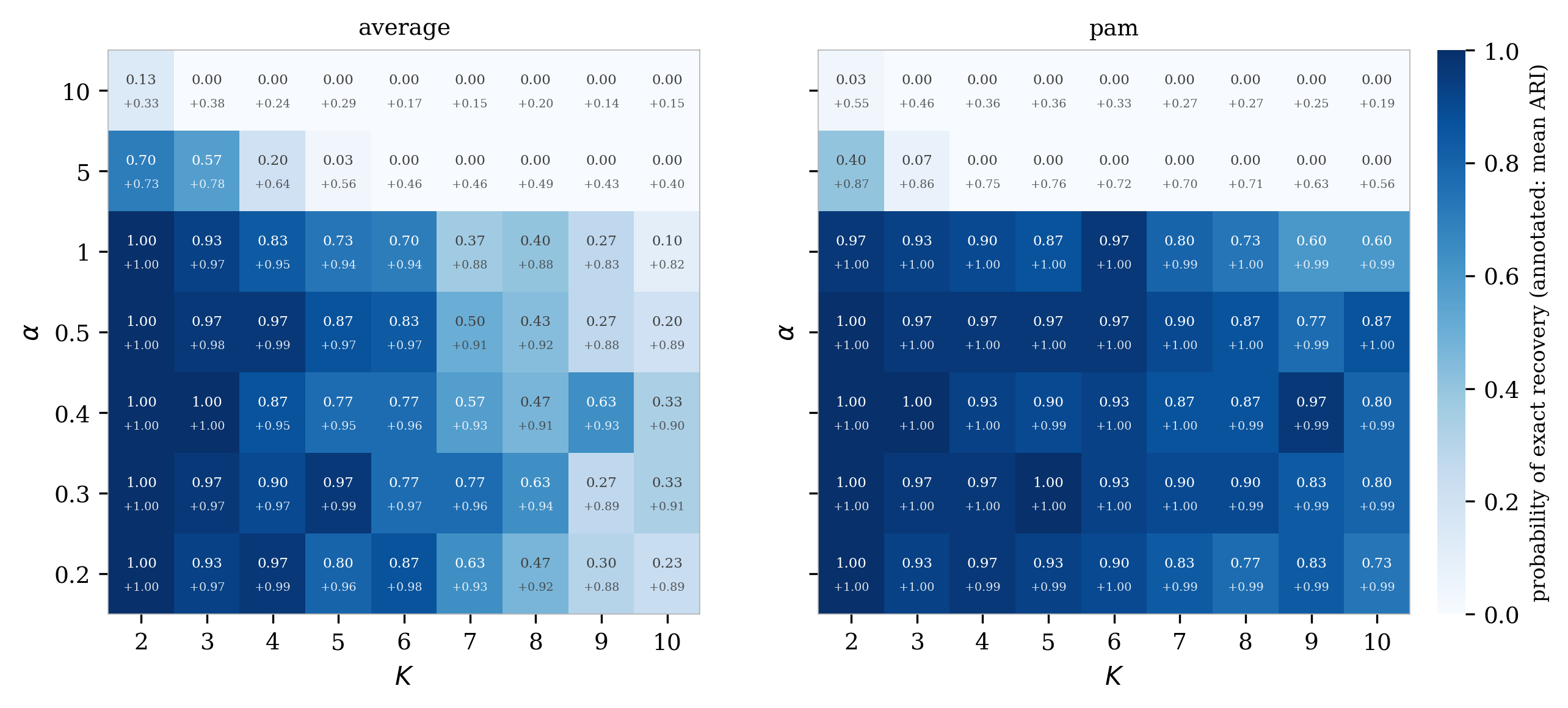}
\caption{Proportion of exact recoveries (above) and mean \textit{ARI}s (below), as functions of $\alpha$ (rows) and $K$ (columns), for Average Linkage and PAM. } 
\label{fig:recovery}
\end{figure}

The recovery region follows the separation region of Figure~\ref{fig:assumption-separation}: for $\alpha\le1$, where separation is established most of the time, the mean ARI never falls below $0.99$ for PAM and $0.82$ for average linkage. On the two top rows, where separation does not hold, both fall away, and at $\alpha=10$ neither exceeds $0.55$. PAM leads on the mean ARI in every cell but two and its advantage is widest on exact recovery inside the separated region ($0.87$ against $0.20$ at $\alpha=0.5$, $K=10$). The order reverses at $\alpha\ge5$: there average linkage recovers exactly more often than PAM ($0.57$ against $0.07$ at $\alpha=5$, $K=3$), on cells where neither reaches an ARI of $0.9$. Single linkage, in Figure~\ref{fig:ARIgridsingle} (appendix), keeps the same shape but performs globally worse: its mean ARI is already $0.47$ at $\alpha=1$, $K=10$, and near zero throughout $\alpha\ge5$.

\subsection{Estimating the number of clusters}
\label{ssec:num-K}

Everything above assumes that the number of clusters $K$ is known. We now apply the profile-graph estimator defined in equation~\eqref{eq:Khat-profile} to the same distance matrices, on the same repetitions and over the same grid. We read it against the average silhouette width, a standard criterion in applied sequence analysis \cite{WeightedCluster}. Figure~\ref{fig:krecovery} reports $\hat{\PP}(\hat K_{N,n}=K)$ and the median $\hat K_{N,n}$ for each $(\alpha,K)$ cell.

\begin{figure}[tbp]
\centering
\includegraphics[width=1\linewidth]{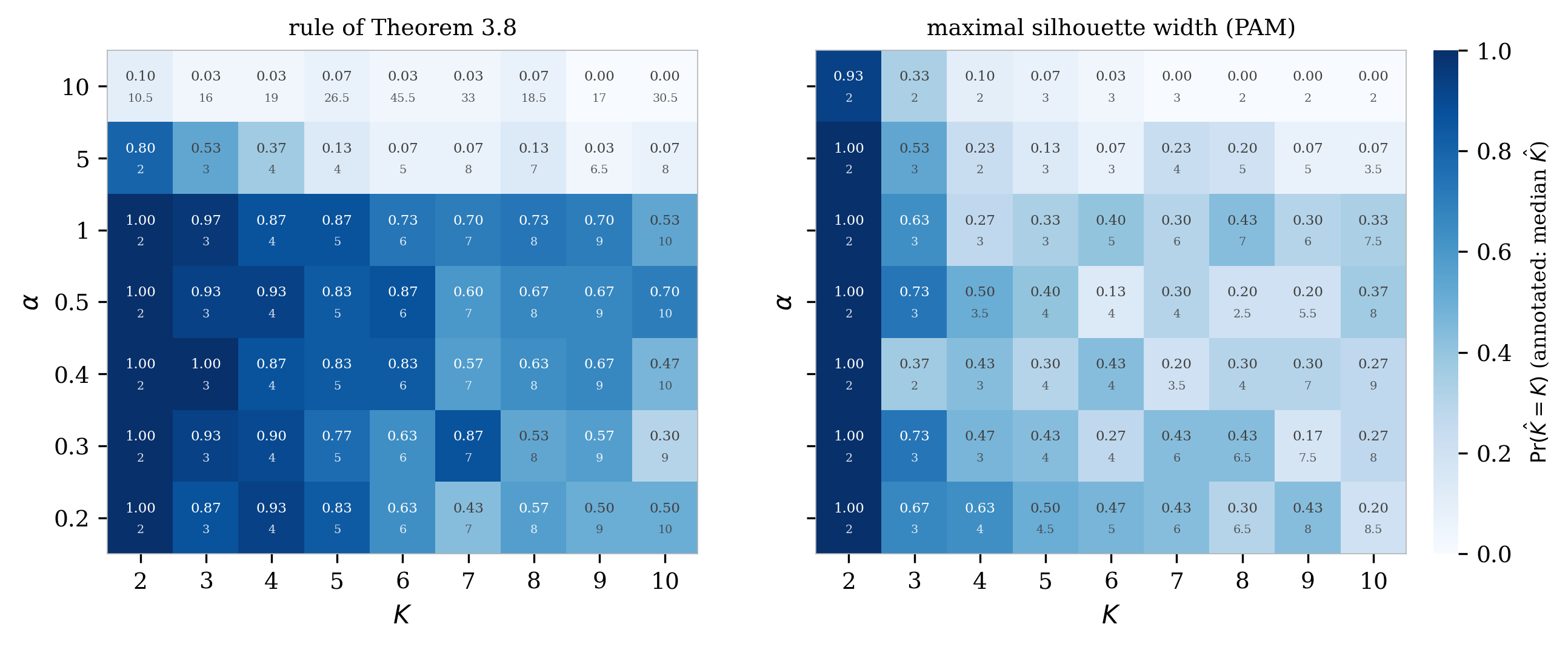}
\caption{Proportion of repetitions with $\hat{K} = K$ using the rule of Theorem~\ref{thm:K-selection} and Average Silhouette Width as functions of $\alpha$ and $K$.}
\label{fig:krecovery}
\end{figure}

The rule finds the correct number of clusters in $75\%$ of the repetitions at $\alpha\le1$, against $45\%$ for the silhouette, and the median $\hat K_{N,n}$ says why: over that whole block the profile rule is centered on the truth, its median equal to $K$ in every cell but one. Though, the silhouette is biased towards small number of clusters, with its median falling below $K$ from $K=4$ on. At $\alpha=10$ the profile rule over-splits, its median rising to $19$ at $K=4$ and to $45.5$ at $K=6$, where the silhouette collapses onto $\hat K_{N,n}=2$ whatever $K$ is.

\subsection{Multichannel hidden Markov models}
\label{ssec:num-hmm}

We repeat the four experiments on the multichannel hidden Markov model of Section~\ref{sec:extensions}. The protocol, the knobs and the reported quantities are those of Sections~\ref{ssec:num-convergence}--\ref{ssec:num-K}. All figures are in Appendix~\ref{app:figures}.

\paragraph{Design.} Each component is a homogeneous HMM with four hidden states emitting five channels of five letters. A single concentration $\alpha$ ties its initial law, hidden transition matrix, and emission matrices in Dirichlet laws. We take $\alpha\in\{1,2,3,5,10,30\}$, $K\in\{2,\ldots,10\}$, $N=800$, $n=1000$, $R=30$ repetitions per cell and $60$ trajectory pairs per entry of $\Gamma^{(n)}$. The ARI paths run at $\alpha=3$ and $K=4$, the convergence of $\hat\gamma_n$ at $\alpha=3$ on the same three cost schemes.

Over the $1620$ repetitions of the grid, no interval places $\eta_n$ below zero, and the median $\hat\eta_n$ is positive in every cell but one: here $\eta$ decreases towards zero from above rather than crossing it, so from $K = 4$, the verdict at $\alpha=30$ is undecided, but never frankly rejected. 

Average linkage now leads. It recovers exactly in $62.2\%$ of the repetitions against $45.5\%$ for PAM, and holds the higher mean ARI as well, $0.87$ against $0.76$. The silhouette leads too: it finds $K$ in $61.9\%$ of the repetitions against $55.2\%$ for the rule of Theorem~\ref{thm:K-selection}, which over-estimates $K$.

\section{Proofs}
\label{sec:proofs}

This section gathers all proofs of the results stated above. We follow the order of the main text.

\subsection{Proofs for Section~\ref{sec:OM}}

\begin{proof}[Proof of Proposition~\ref{prop:convergence}]
\emph{Step 1: the stationary case.} We first prove the convergence for two independent stationary ergodic processes $X,Y$ on $\Sigma$, which covers the stationary case $(\mu,\mu')=(\pi,\pi')$. Work on the space $\Omega:=\Sigma^{\N}\times\Sigma^{\N}$, endowed with $\mu_X\otimes\nu_Y$, where $\mu_X:=\mathcal L(X)$ and $\nu_Y:=\mathcal L(Y)$. Let $S(x_1,x_2,\ldots):=(x_2,x_3,\ldots)$ be the shift operator on $\Sigma^{\N}$ and set $T:=S\times S$. Stationarity implies that $T$ preserves $\mu_X\otimes\nu_Y$. 

Concatenating an optimal alignment of the first $n$ symbols with an optimal alignment of the following $m$ symbols yields an alignment of the first $n+m$ symbols. Its cost is $D_n+D_m\circ T^n$, hence
\begin{equation}
\label{eq:subadd}
D_{n+m} \;\le\; D_n + D_m\circ T^n .
\end{equation}
Moreover $0\le D_n\le Mn$. Kingman's subadditive ergodic theorem~\cite{kingman} yields the almost sure and $L^1$ convergence of $\hat\gamma_n$ to a limit $g$ with $\E[g]=\inf_{n\ge1}\E[D_n]/n$.

It remains to see that $g$ is deterministic. Set $\Delta:=\max_{a\in\Sigma}\delta(a)$ and fix $k\ge1$. Deleting $y_1,\ldots,y_k$ and then inserting $y_{n+1},\ldots,y_{n+k}$ transforms $y_{1:n}$ into $y_{k+1:k+n}$, so
\[
\dom(y_{1:n},y_{k+1:k+n})\le 2k\Delta.
\]
The same construction applied to the first argument gives
\[
\dom(x_{1:n},x_{k+1:k+n})\le 2k\Delta.
\]
By the triangle inequality for $\dom$, for all $x,y$ and all $n\ge k$,
\[
  \bigl|D_n(x,y)-D_n(x,S^ky)\bigr| \le 2k\Delta,
  \qquad
  \bigl|D_n(x,y)-D_n(S^kx,y)\bigr| \le 2k\Delta .
\] 

Define $\displaystyle{g^\star(x,y) := \limsup_{n\to\infty}\frac{D_n(x,y)}{n}}$. By Kingman's theorem, $g^\star=g$ almost surely. Dividing the preceding inequalities by $n$ and taking the limsup yields
\[
g^\star(x,S^ky)=g^\star(x,y),
\qquad
g^\star(S^kx,y)=g^\star(x,y).
\]

By Fubini and ergodicity of $\mu_X$, for $\nu_Y$-almost every $y$ the map $x\mapsto g^\star(x,y)$ is $\mu_X$-almost surely constant, equal to
\[
h(y):=\int g^\star(x,y)\,\mu_X(dx).
\]
The invariance under $\mathrm{id}\times S$ then gives $h\circ S=h$ $\nu_Y$-almost everywhere, so $h$ is constant by ergodicity of $\nu_Y$. Since $g^\star=g$ almost surely, $g$ is deterministic.

Applied to two independent stationary chains with kernels $P$ and $Q$, which are ergodic under the shift by irreducibility, this gives $\hat\gamma_n\to\gamma(P,Q)=\inf_{n\ge1}\E_{\pi,\pi'}[D_n]/n$, $\PP_{\pi,\pi'}$-almost surely.

\emph{Step 2: change of initial law.} We show that the convergence obtained in Step~1 for $(\mu,\mu')=(\pi,\pi')$ transfers to arbitrary initial laws. Since $\Sigma$ is finite and $P$ is irreducible, $\pi$ has full support, so $\rho:\Sigma^\N\to\R^+$, $\rho(x):=\mu(x_1)/\pi(x_1)$, is well defined and bounded.

For $n\ge1$ and $a_{1:n}\in\Sigma^n$, let $C(a_{1:n}):=\{x\in\Sigma^\N: x_{1:n}=a_{1:n}\}$ denote the corresponding cylinder. By the Markov property and time-homogeneity of the kernel,
\[
  \PP_\mu\bigl(C(a_{1:n})\bigr) = \mu(a_1)\prod_{t=1}^{n-1}P(a_t,a_{t+1}),
  \qquad
  \PP_\pi\bigl(C(a_{1:n})\bigr) = \pi(a_1)\prod_{t=1}^{n-1}P(a_t,a_{t+1}).
\]
On $C(a_{1:n})$ the function $\rho$ is constant, equal to $\mu(a_1)/\pi(a_1)$, so
\begin{equation}
\label{eq:cyl}
  \E_\pi\bigl[\rho\,\ind_{C(a_{1:n})}\bigr]
  = \frac{\mu(a_1)}{\pi(a_1)}\,\PP_\pi\bigl(C(a_{1:n})\bigr)
  = \mu(a_1)\prod_{t=1}^{n-1}P(a_t,a_{t+1})
  = \PP_\mu\bigl(C(a_{1:n})\bigr).
\end{equation}

The weight $\rho$ is the same for all lengths $n$: it is $\sigma(X_1)$-measurable, hence independent of the horizon.

Let now
\[
  \mathcal D := \bigl\{ A\in\mathcal F : \PP_\mu(A)=\E_\pi[\rho\,\ind_A]\bigr\},
\]
where $\mathcal F$ is the product $\sigma$-algebra on $\Sigma^\N$, with $\Sigma$ equipped with the discrete $\sigma$-algebra. The cylinders, together with $\emptyset$ and $\Sigma^\N$, form a $\pi$-system generating $\mathcal F$. By \eqref{eq:cyl} this $\pi$-system is contained in $\mathcal D$. Moreover, $\mathcal D$ is a $\lambda$-system. Dynkin's lemma therefore gives $\mathcal D=\mathcal F$, that is, $\PP_\mu\ll\PP_\pi$ on the whole of $\mathcal F$, with $d\PP_\mu/d\PP_\pi=\rho$.

Since $X$ and $Y$ are independent under both $\PP_{\mu,\mu'}$ and $\PP_{\pi,\pi'}$, the same holds for the product laws: $\PP_{\mu,\mu'}\ll\PP_{\pi,\pi'}$ on the product $\sigma$-algebra of $\Sigma^\N\times\Sigma^\N$, with density $\rho\otimes\rho'$.

Each $D_n$ is a measurable function of $(X_{1:n},Y_{1:n})$, so $A:=\{\hat\gamma_n\to\gamma(P,Q)\}$ is measurable for that $\sigma$-algebra. By Step~1, $\PP_{\pi,\pi'}(A^c)=0$, so $(\rho\otimes\rho')\ind_{A^c}$ vanishes $\PP_{\pi,\pi'}$-almost surely and
\[
  \PP_{\mu,\mu'}(A^c)
  = \E_{\pi,\pi'}\bigl[(\rho\otimes\rho')\,\ind_{A^c}\bigr]
  = 0 .
\]
Hence $\hat\gamma_n\to\gamma(P,Q)$ $\PP_{\mu,\mu'}$-almost surely, with the same limit as in the stationary case.
\end{proof}

\begin{proof}[Proof of Proposition~\ref{prop:concentration}]
The pair $Z=(X,Y)$ is an irreducible aperiodic Markov chain on $\Sigma^2$. Define $f:(\Sigma^2)^n\to\R$ by
\[
f\bigl((x_1,y_1),\ldots,(x_n,y_n)\bigr) := \frac{1}{n}\dom(x_{1:n},y_{1:n}).
\]

We show that $f$ has bounded differences $c_i=2M/n$ in each coordinate, i.e., for every $z,z'\in(\Sigma^2)^n$,
\begin{equation}
\label{eq:bd}
|f(z)-f(z')| \;\le\; \sum_{i=1}^n c_i\,\ind_{z_i\neq z'_i},\qquad c_i=\frac{2M}{n}.
\end{equation}
Write $z=((x_t,y_t))_{t\le n}$ and $z'=((x'_t,y'_t))_{t\le n}$. The triangle inequality for $\dom$ gives
\[
\bigl|\dom(x,y)-\dom(x',y')\bigr| \;\le\; \dom(x,x') + \dom(y,y').
\]
Substituting each mismatched pair yields $\dom(x,x')\le\sum_{t:x_t\neq x'_t}c_{\sub}(x_t,x'_t)\le M\,\dH(x,x')$, where $\dH$ is the Hamming distance, and similarly for $\dom(y,y')$. Hence
\[
n\,|f(z)-f(z')| \;\le\; M\bigl(\dH(x,x')+\dH(y,y')\bigr) \;\le\; 2M\,\dH(z,z'),
\]
which is~\eqref{eq:bd}.

Paulin's concentration inequality obtained by Marton coupling~\cite{paulin} provides a universal constant $C_{\mathrm{Pau}}>0$ such that, for every function $g$ on $(\Sigma^2)^n$ with bounded differences $(c_i)_{i\le n}$ in the sense of~\eqref{eq:bd} and every $s\ge 0$,
\begin{equation}
\label{eq:paulin-generic}
\PP\bigl(|g(Z_{1:n})-\E[g(Z_{1:n})]|\ge s\bigr) \;\le\; 2\exp\!\Bigl(-\frac{8s^2}{C_{\mathrm{Pau}}\,\tmix\sum_{i=1}^n c_i^2}\Bigr).
\end{equation}
Applying this to $f$ with $\sum_{i=1}^n c_i^2 = n\cdot(2M/n)^2 = 4M^2/n$,
\[
\PP\bigl(|\hat\gamma_n-\E\hat\gamma_n|\ge s\bigr) \;\le\; 2\exp\!\Bigl(-\frac{2s^2 n}{C_{\mathrm{Pau}}M^2\tmix}\Bigr),
\]
which is the desired result.
\end{proof}

\begin{proof}[Proof of Lemma~\ref{lem:OM-score-properties}]
By Assumption~\ref{ass:metric}(iv),
\[
s(a,b)=\delta(a)+\delta(b)-c_{\sub}(a,b)\ge0,
\]
and $s$ is symmetric by Assumption~\ref{ass:metric}(ii).

For an alignment of $x_{1:p}$ and $y_{1:q}$, let
\[
(i_1,j_1),\ldots,(i_r,j_r),
\qquad
i_1<\cdots<i_r,\quad
j_1<\cdots<j_r,
\]
be its letter--letter matched pairs. All remaining letters are matched to gaps. Hence its cost is
\begin{align*}
&
\sum_{i\notin\{i_1,\ldots,i_r\}}\delta(x_i)
+
\sum_{j\notin\{j_1,\ldots,j_r\}}\delta(y_j)
+
\sum_{h=1}^r c_{\sub}(x_{i_h},y_{j_h})
\\
&\qquad =
\sum_{i=1}^p\delta(x_i)
+
\sum_{j=1}^q\delta(y_j)
-
\sum_{h=1}^r s(x_{i_h},y_{j_h}).
\end{align*}

Conversely, any increasing set of matched pairs defines an admissible alignment. Hence minimizing the alignment cost is equivalent to maximizing the total score, which gives~\eqref{eq:OM-score-identity}.

For superadditivity, take optimal matchings for $(x_{1:p},y_{1:q})$ and $(x_{p+1:p+p'},y_{q+1:q+q'})$. Shifting the indices of the second
matching by $(p,q)$ and concatenating the two sets of matched pairs gives an admissible increasing matching for
$(x_{1:p+p'},y_{1:q+q'})$. Therefore
\[
L_s(x_{1:p+p'},y_{1:q+q'})
\ge
L_s(x_{1:p},y_{1:q})
+
L_s(x_{p+1:p+p'},y_{q+1:q+q'}).
\]
Likewise, any matching remains admissible when either word is extended, so $L_s$ cannot decrease under extension.

Since the empty matching is admissible and $s\ge0$, we have $L_s\ge0$.
An increasing matching contains at most $\min(p,q)$ pairs, each contributing at most $F_s$, hence
\[
L_s(x_{1:p},y_{1:q})
\le
F_s\min(p,q).
\]

Suppose now that $x,x'\in\Sigma^p$ differ only at position $i$. Take an optimal matching for $(x,y)$ and use the same matched index pairs for $(x',y)$. If $i$ is unmatched, the score is unchanged; if it is matched to some $j$, only one summand changes, by at most $A_s$. Thus
\[
L_s(x',y)\ge L_s(x,y)-A_s.
\]
Interchanging $x$ and $x'$ gives
\[
|L_s(x,y)-L_s(x',y)|\le A_s.
\]
By symmetry of $s$, the same argument applies to a one-letter change in
$y$.

It remains to prove~\eqref{eq:OM-rectangular-block}. Assume first that
$p+q=2n$. Superadditivity gives
\[
L_s(X_{1:2n},Y_{1:2n})
\ge
L_s(X_{1:p},Y_{1:q})
+
L_s(X_{p+1:2n},Y_{q+1:2n}).
\]
The two blocks in the second term have respective lengths
$2n-p=q$ and $2n-q=p$. By stationarity and independence,
\[
\E L_s(X_{p+1:2n},Y_{q+1:2n})
=
\E L_s(X_{1:q},Y_{1:p}).
\]
Since $X$ and $Y$ are independent copies of the same process and $s$ is symmetric,
\[
\E L_s(X_{1:q},Y_{1:p})
=
\E L_s(X_{1:p},Y_{1:q}).
\]
Taking expectations in the superadditivity inequality yields
\[
\E L_s(X_{1:p},Y_{1:q})
\le
\frac12\,\E L_s(X_{1:2n},Y_{1:2n}).
\]

Finally, if $p+q<2n$, choose $p'\ge p$ and $q'\ge q$ such that
$p'+q'=2n$. By monotonicity under extension,
\[
L_s(X_{1:p},Y_{1:q})
\le
L_s(X_{1:p'},Y_{1:q'}),
\]
and the preceding case applied to $(p',q')$ proves~\eqref{eq:OM-rectangular-block}.
\end{proof}

\begin{proof}[Proof of Proposition~\ref{prop:diagonal-rate}]
Let $\pi$ be the stationary distribution of $P$. We first isolate the rate argument from the proof of \cite[Theorem~3.1]{houdre2019}, specialized to the independent-identical case of \cite[Corollary~3.1]{houdre2019}.

\medskip
\noindent\textbf{Technical lemma.}
Let $X$ and $Y$ be independent stationary copies of the Markov chain with
transition matrix $P$, and let
\[
L_n:=L_s(X_{1:n},Y_{1:n}).
\]
Then, there exists a constant $C_s=C_s(P,s)>0$ such that, for every
$n\ge2$,
\begin{equation}
\label{eq:score-rate-stationary}
0
\le
\lambda_s(P)-\frac{\E[L_n]}{n}
\le
C_s\sqrt{\frac{\log n}{n}},
\end{equation}
where
\[
\lambda_s(P)
:=
\lim_{m\to\infty}\frac{\E[L_m]}{m}.
\]

\smallskip
\noindent\emph{Proof of the technical lemma.}
By Lemma~\ref{lem:OM-score-properties}(ii) and stationarity,
$(\E[L_n])_{n\ge1}$ is superadditive, so Fekete's lemma gives
\[
\lambda_s(P)=\sup_{n\ge1}\frac{\E[L_n]}{n},
\]
and hence the lower bound. For the upper bound, apply the partition argument in the proof of \cite[Theorem~3.1]{houdre2019}, specialized to the independent-identical case as in \cite[Corollary~3.1]{houdre2019}. Lemma~\ref{lem:OM-score-properties}(v), (iv), and (iii) provide, respectively, the required block comparison, bounded-difference estimate, and linear bound for $L_s$. The partition construction and its counting are unchanged. Hence, their argument yields
\[
\lambda_s(P)-\frac{\E[L_n]}{n}
\le C_s\sqrt{\frac{\log n}{n}},
\]
with only a change of constants.
\hfill$\diamond$

\medskip
We now return to the OM distance. Taking expectations in \eqref{eq:OM-score-identity} under stationarity gives
\[
\frac1n\E_{\pi,\pi}[D_n]
=
2\E_\pi[\delta(X_1)]
-
\frac1n\E_{\pi,\pi}
\bigl[L_s(X_{1:n},Y_{1:n})\bigr].
\]
By Proposition~\ref{prop:convergence}, the left-hand side converges to $\gamma(P,P)$. Hence
\[
\gamma(P,P)
=
2\E_\pi[\delta(X_1)]-\lambda_s(P).
\]
Subtracting this identity from the preceding finite-$n$ identity and using \eqref{eq:score-rate-stationary} yields
\begin{equation}
\label{eq:stationary-diagonal-rate}
0
\le
\frac1n\E_{\pi,\pi}[D_n]-\gamma(P,P)
\le
C_s\sqrt{\frac{\log n}{n}}.
\end{equation}

It remains to generalize this rate for arbitrary initial distributions. We use the coupling argument of \cite[Remark~2.1(ii)]{houdre2019}, also invoked in the proof of \cite[Corollary~3.1]{houdre2019}. Coupling each chain started from $\mu$ with a stationary copy yields constants $c>0$ and $\alpha\in(0,1)$, depending only on $P$, such that for every $K\ge0$,
\[
\sup_\mu \PP(\tau>K)\le c\alpha^K,
\]
where $\tau$ is the maximum of the two coupling times.

On $\{\tau\le K\}$, the two pairs of trajectories may differ only among their first $K$ coordinates. Aligning each trajectory with its stationary counterpart position by position gives a cost of at most $MK$; hence, by the triangle inequality for $\dom$,
\[
|D_n(X,Y)-D_n(\bar X,\bar Y)|\le 2MK.
\]
Since both distances are bounded by $Mn$, the same argument as in \cite[Eq.~(2.5)]{houdre2019} gives
\[
\left|
\frac1n\E_{\mu,\mu}[D_n]
-
\frac1n\E_{\pi,\pi}[D_n]
\right|
\le
cM\alpha^K+\frac{2MK}{n}.
\]
Taking $K=\lceil\sqrt n\rceil$ and combining with \eqref{eq:stationary-diagonal-rate} proves~\eqref{eq:diagonal-rate}, after adjusting the constant. The bound is uniform in $\mu$.
\end{proof}

\begin{proof}[Proof of Proposition~\ref{prop:usc}]
For $n\ge1$, set
\[
g_n(P,Q) := \frac{1}{n}\,\E_{P,Q}\bigl[\dom(X_{1:n},Y_{1:n})\bigr],
\]
the expectation being taken under the law of two independent stationary Markov chains with transitions $P$ and $Q$, so that $\gamma=\inf_{n\ge1}g_n$ by Proposition~\ref{prop:convergence}.

Each $g_n$ is continuous on $\mathcal S^\circ_d\times\mathcal S^\circ_d$. Indeed, on $\mathcal S^\circ_d$ the linear system $\pi_P P=\pi_P$, $\sum_a\pi_P(a)=1$ uniquely determines $\pi_P$, and Cramer's rule expresses $\pi_P$ as a rational function of the entries of $P$ with non-vanishing denominator on $\mathcal S^\circ_d$. Hence $P\mapsto\pi_P$ is continuous, and for each fixed $x_{1:n}\in\Sigma^n$,
\[
\mu_P^{(n)}(x_{1:n}) = \pi_P(x_1)\prod_{t=1}^{n-1}P(x_t,x_{t+1})
\]
is continuous in $P$. Since $\Sigma^n$ is finite,
\[
g_n(P,Q) = \frac{1}{n}\sum_{x,y\in\Sigma^n}\dom(x,y)\,\mu_P^{(n)}(x)\,\mu_Q^{(n)}(y)
\]
is a finite sum of continuous functions of $(P,Q)$, hence continuous.

Thus, $\gamma = \inf_{n\ge1} g_n$ is the pointwise infimum of a family of continuous functions, and any such infimum is upper semi-continuous. Explicitly, for $(P_m,Q_m)\to(P_0,Q_0)$ and any $\varepsilon>0$, choose $n_0$ such that $g_{n_0}(P_0,Q_0)<\gamma(P_0,Q_0)+\varepsilon$. By continuity of $g_{n_0}$, $g_{n_0}(P_m,Q_m)\to g_{n_0}(P_0,Q_0)$, and since $\gamma(P_m,Q_m)\le g_{n_0}(P_m,Q_m)$,
\[
\limsup_{m\to\infty}\gamma(P_m,Q_m)\le g_{n_0}(P_0,Q_0)<\gamma(P_0,Q_0)+\varepsilon.
\]
Letting $\varepsilon\to0$ concludes.
\end{proof}

\begin{proof}[Proof of Proposition~\ref{prop:bounds}]

By Proposition~\ref{prop:convergence}, $\gamma(P,Q)$ does not depend on the initial laws; we may therefore take $X$ and $Y$ stationary throughout.

\emph{Lower bound.} $(\Sigma_\bot,\bar d)$ is a metric space. Let $f:\Sigma_\bot\to\R$ be $1$-Lipschitz with respect to $\bar d$. Fix $n\ge 1$ and an alignment $(\tilde x_{1:L},\tilde y_{1:L})\in\Sigma_\bot^L\times\Sigma_\bot^L$ of $(X_{1:n},Y_{1:n})$. By the $1$-Lipschitz property and the triangle inequality,
\[
\mathrm{cost}(\tilde x,\tilde y) = \sum_{\ell=1}^L \bar d(\tilde x_\ell,\tilde y_\ell) \;\ge\; \sum_{\ell=1}^L |f(\tilde x_\ell)-f(\tilde y_\ell)| \;\ge\; \Bigl|\sum_{\ell=1}^L f(\tilde x_\ell) - \sum_{\ell=1}^L f(\tilde y_\ell)\Bigr|.
\]
The aligned sequence $\tilde x$ contains the $n$ original symbols of $X_{1:n}$ together with $L-n$ copies of $\bot$; the same holds for $\tilde y$. The $f(\bot)$ terms cancel between the two sums, and we obtain
\[
\mathrm{cost}(\tilde x,\tilde y) \;\ge\; \Bigl|\sum_{t=1}^n f(X_t) - \sum_{t=1}^n f(Y_t)\Bigr|.
\]
The right-hand side does not depend on the alignment, so taking the minimum over alignments,
\[
\dom(X_{1:n},Y_{1:n}) \;\ge\; \Bigl|\sum_{t=1}^n f(X_t) - \sum_{t=1}^n f(Y_t)\Bigr|.
\]
Taking expectations, dividing by $n$, and applying Jensen's inequality to bring the expectation inside the absolute value, then using stationarity of $X,Y$:
\[
\frac{\E[D_n]}{n} \;\ge\; \E\Bigl|\tfrac{1}{n}\sum_{t=1}^n f(X_t) - \tfrac{1}{n}\sum_{t=1}^n f(Y_t)\Bigr| \;\ge\; \bigl|\E[f(X_1)]-\E[f(Y_1)]\bigr|.
\]
Letting $n\to\infty$, $\gamma(P,Q)\ge|\E_{\pi_P}[f]-\E_{\pi_Q}[f]|$. Taking the supremum over all $1$-Lipschitz $f:\Sigma_\bot\to\R$ and applying Kantorovich--Rubinstein duality on the finite metric space $(\Sigma_\bot,\bar d)$ yields $\gamma(P,Q)\ge\Wd(\pi_P,\pi_Q)$.

\emph{Upper bound.} Aligning $X_{1:n}$ and $Y_{1:n}$ position by position, with no gaps, gives an alignment of cost $\sum_{t=1}^n c_{\sub}(X_t,Y_t)$. Hence $\dom(X_{1:n},Y_{1:n})\le\sum_{t=1}^n c_{\sub}(X_t,Y_t)$. Dividing by $n$, taking expectations, and using stationarity and the independence of $X$ and $Y$:
\[
\frac{\E[D_n]}{n} \;\le\; \E[c_{\sub}(X_1,Y_1)] \;=\; \sum_{a,b\in\Sigma}\pi_P(a)\pi_Q(b)c_{\sub}(a,b) \;=\; \pi_P^\top S\,\pi_Q.
\]
Letting $n\to\infty$ gives the claim.
\end{proof}

\subsection{Proofs for Section~\ref{sec:clustering}}

\begin{proof}[Proof of Lemma~\ref{lem:uniform}]
For each pair $(i,j)$ with $i<j$, conditionally on $(Z_i,Z_j)=(k,k')$, the sequences $X_i$ and $X_j$ are independent realizations of Markov chains with transition matrices $P_k$ and $P_{k'}$. By Proposition~\ref{prop:concentration},
\[
\PP\left(
\left|
\hat\gamma_n(i,j)-\Gamma^{(n)}_{kk'}
\right|>\varepsilon
\,\middle|\,
Z_i=k,Z_j=k'
\right)
\le
2\exp\left(
-\frac{2\varepsilon^2 n}
{C(M,P_k,P_{k'})}
\right).
\]

Since $C(M,P_k,P_{k'})\le C^\star$, marginalizing over $(Z_i,Z_j)$ gives

\[
\PP\left(
\left| \hat\gamma_n(i,j)-\Gamma^{(n)}_{Z_iZ_j} \right| > \varepsilon
\right) \le 2\exp\left( -\frac{2\varepsilon^2 n}{C^\star}
\right).
\]

A union bound over the $\binom{N}{2}$ pairs therefore yields

\[
\PP\bigl(\mathcal E_{N,n}(\varepsilon)^c\bigr)
\le
\binom{N}{2}
\,2\exp\left(
-\frac{2\varepsilon^2 n}{C^\star}
\right)
\le
N^2\exp\left(
-\frac{2\varepsilon^2 n}{C^\star}
\right).
\]

Finally, under Assumption~\ref{ass:growth}, $2\log N=o(n)$, so the right-hand side converges to zero.
\end{proof}

\begin{proof}[Proof of Theorem~\ref{thm:SL}]
For $k=k'$, Proposition~\ref{prop:diagonal-rate} yields, uniformly over $k$,
\[
\Gamma^{(n)}_{kk} \le \Gamma_{kk} + B_1\sqrt{\frac{\log n}{n}}
\]
for some $B_1 \in \R_+$.

For $k\neq k'$, the finite-horizon coupling argument used in the proof of Proposition~\ref{prop:diagonal-rate} gives constants $c_{kk'} \in \R_+$ and $\alpha_{kk'}\in(0,1)$ such that
\[
\left|\Gamma^{(n)}_{kk'} - \frac1n\E_{\pi_k,\pi_{k'}}[D_n] \right|
\le c_{kk'}\left( \frac1{\sqrt n} + \alpha_{kk'}^{\sqrt n}
\right).
\]

Moreover, Proposition~\ref{prop:convergence} gives
\[
\frac1n\E_{\pi_k,\pi_{k'}}[D_n] \ge \Gamma_{kk'}.
\]

Since $K$ is fixed, there exists $B_2 \in \R_+$ such that, uniformly over $k\neq k'$ and $n\ge2$,
\[
\Gamma^{(n)}_{kk'} \ge \Gamma_{kk'} - B_2\sqrt{\frac{\log n}{n}}.
\]

Set $\displaystyle{B:=\max(B_1,B_2)}$ and $\displaystyle{r_n:=B\sqrt{\frac{\log n}{n}}}.$
On $\mathcal E_{N,n}(\varepsilon)$,
\[
Z_i=Z_j \quad\Longrightarrow\quad \hat\gamma_n(i,j) \le \gin+r_n+\varepsilon,
\]
\[
Z_i\neq Z_j \quad\Longrightarrow\quad \hat\gamma_n(i,j) \ge \gout-r_n-\varepsilon.
\]

By assumption,
\[
\gin+r_n+\varepsilon < \gout-r_n-\varepsilon.
\]
Choose any $t_n\in \bigl( \gin+r_n+\varepsilon,\, \gout-r_n-\varepsilon \bigr)$.

Then, on $\mathcal E_{N,n}(\varepsilon)$, every non-empty $G_k$ induces a complete subgraph of $\mathcal G_{t_n}$, while there is no edge between two distinct classes. Hence, if every $G_k$ is non-empty, the connected components of $\mathcal G_{t_n}$ are exactly $G_1,\ldots,G_K$. Since the single-linkage partition at level $t_n$ is the set of connected components of $\mathcal G_{t_n}$, its cut at $K$ blocks is therefore $\mathcal P^\star$.

Consequently,
\[
\PP\!\left(
\widehat{\mathcal P}^{\mathrm{SL}}_{N,n}(K)
\neq
\mathcal P^\star
\right)
\le
\PP\bigl(\mathcal E_{N,n}(\varepsilon)^c\bigr)
+
\PP\bigl(\exists k:G_k=\emptyset\bigr).
\]
By Lemma~\ref{lem:uniform},
\[
\PP\bigl(\mathcal E_{N,n}(\varepsilon)^c\bigr)
\le
N^2\exp\!\left(
-\frac{2\varepsilon^2n}{C^\star}
\right),
\]
while
\[
\PP\bigl(\exists k:G_k=\emptyset\bigr)
\le
\sum_{k=1}^K(1-w_k)^N
\le
K(1-w_{\min})^N.
\]
This proves the finite-sample bound. Under Assumption~\ref{ass:growth}, the two terms on the right-hand side converge to zero and $r_n\to0$, proving the consistency statement.
\end{proof}

\begin{proof}[Proof of Lemma~\ref{lem:kmedoid-one-per-cluster}]
Let $\mathcal M:=\widehat{\mathcal M}^{\mathrm{PAM}}_{N,n}$ be the set of medoids returned by PAM, with $N\ge\max\{N_0,K\}$, and define $\Phi_i(\mathcal M):=\min_{m\in\mathcal M}\hat\gamma_n(i,m)$. By construction of PAM,
\begin{equation}
\label{eq:pam-local}
\Phi_{N,n}(\mathcal M)\le \Phi_{N,n}\bigl((\mathcal M\setminus\{m\})\cup\{j\}\bigr)
\end{equation}
for every $m\in\mathcal M$ and $j\notin\mathcal M$.

Set
\[
\varepsilon_n:=r-b_n,\qquad u:=\gin+r,\qquad v:=\gout-r.
\]
Since $r<w_{\min}\eta/16$ and $w_{\min}\le1/2$, we have $r<\eta/32$, and therefore
\begin{equation}
\label{eq:pam-gap}
v-u=\eta-2r\ge\frac{\eta}{2}>0.
\end{equation}
On $\mathcal E_{N,n}(\varepsilon_n)$, the bounds established in the proof of Theorem~\ref{thm:SL} give, for $i\neq j$,
\begin{equation}
\label{eq:pam-bounds}
\hat\gamma_n(i,j)\le u\quad\text{if }Z_i=Z_j,\qquad
\hat\gamma_n(i,j)\ge v\quad\text{if }Z_i\neq Z_j.
\end{equation}
Moreover, on $\mathcal W_N(\delta_0)$,
\begin{equation}
\label{eq:pam-weight}
\hat w_\ell\ge w_{\min}-\delta_0=\frac{w_{\min}}{2}\qquad\text{for all }\ell.
\end{equation}

Assume, for contradiction, that  a cluster $G_k$ contains no medoid of $\mathcal M$. Since $\mathcal M$ contains $K$ medoids distributed among the remaining $K-1$ clusters, by the pigeonhole principle there exist $m_1\neq m_2$ in $\mathcal M$ such that $Z_{m_1}=Z_{m_2}=:k'\neq k$. Fix any $j\in G_k$ and set
\[
\mathcal M':=(\mathcal M\setminus\{m_2\})\cup\{j\}.
\]
Since $G_k\cap\mathcal M=\emptyset$, this is an admissible one-medoid swap.

Let $\Psi_i:=\min_{m\in\mathcal M\setminus\{m_2\}}\hat\gamma_n(i,m)$. Since $\mathcal M\setminus\{m_2\}\subseteq\mathcal M'$, we have $\Phi_i(\mathcal M')\le\Psi_i$, while $\Phi_i(\mathcal M)=\min(\Psi_i,\hat\gamma_n(i,m_2))$. Moreover, $m_1\in\mathcal M\setminus\{m_2\}$ gives $\Psi_i\le\hat\gamma_n(i,m_1)$. Hence
\begin{align}
\Phi_i(\mathcal M)-\Phi_i(\mathcal M')
&\ge \min\bigl(0,\hat\gamma_n(i,m_2)-\Psi_i\bigr)\notag\\
&\ge \min\bigl(0,\hat\gamma_n(i,m_2)-\hat\gamma_n(i,m_1)\bigr).
\label{eq:pam-master}
\end{align}

If $i\notin\{m_1,m_2\}$, then $\hat\gamma_n(i,m_1)$ and $\hat\gamma_n(i,m_2)$ are both within $\varepsilon_n$ of the same finite-horizon mean $\Gamma^{(n)}_{Z_i k'}$. Therefore
\begin{equation}
\label{eq:pam-baseline}
\Phi_i(\mathcal M)-\Phi_i(\mathcal M')\ge-2\varepsilon_n.
\end{equation}

If $i\in G_k$, then every medoid in $\mathcal M$ belongs to another cluster, so $\Phi_i(\mathcal M)\ge v$. On the other hand, $j\in\mathcal M'\cap G_k$, hence $\Phi_i(\mathcal M')\le u$ if $i\neq j$, while $\Phi_j(\mathcal M')=0$. Thus, in both cases,
\begin{equation}
\label{eq:pam-gain}
\Phi_i(\mathcal M)-\Phi_i(\mathcal M')\ge v-u.
\end{equation}

Finally, for $i=m_1$ we have $\Phi_{m_1}(\mathcal M)=\Phi_{m_1}(\mathcal M')=0$, whereas for $i=m_2$,
\[
\Phi_{m_2}(\mathcal M)-\Phi_{m_2}(\mathcal M')\ge-u.
\]

Summing these bounds and dividing by $N$ gives
\begin{align}
\Phi_{N,n}(\mathcal M)-\Phi_{N,n}(\mathcal M')
&\ge -2\varepsilon_n+\hat w_k(v-u)-\frac{u}{N}\notag\\
&>-\frac{w_{\min}\eta}{8}+\frac{w_{\min}\eta}{4}-\frac{w_{\min}\eta}{16}\notag\\
&=\frac{w_{\min}\eta}{16}>0.
\label{eq:pam-swap-gain}
\end{align}
Indeed, $2\varepsilon_n<2r<w_{\min}\eta/8$, \eqref{eq:pam-weight} and \eqref{eq:pam-gap} give $\hat w_k(v-u)\ge w_{\min}\eta/4$, and $N\ge N_0$ gives $u/N\le w_{\min}\eta/16$.

The swap from $\mathcal M$ to $\mathcal M'$ strictly decreases the objective, contradicting~\eqref{eq:pam-local}. Hence every cluster contains at least one medoid. Since there are $K$ clusters and $|\mathcal M|=K$, each cluster contains exactly one medoid.
\end{proof}

\begin{proof}[Proof of Theorem~\ref{thm:kmedoids}]
Set
\[
b_n:=B\sqrt{\frac{\log n}{n}},\qquad \varepsilon_n:=r-b_n,
\]
and work on $\mathcal E_{N,n}(\varepsilon_n)\cap\mathcal W_N(\delta_0)$ with $N\ge\max\{N_0,K\}$, keeping the notation $u=\gin+r$ and $v=\gout-r$ of the previous proof. By Lemma~\ref{lem:kmedoid-one-per-cluster}, the PAM output $\mathcal M:=\widehat{\mathcal M}^{\mathrm{PAM}}_{N,n}$ has exactly one medoid in each class. In that case, for any $i\in\{1,\ldots,N\}$:
\begin{itemize}[leftmargin=*,topsep=2pt]
\item the unique medoid $m\in\mathcal M\cap G_{Z_i}$ satisfies $\hat\gamma_n(i,m)\le u$ by~\eqref{eq:pam-bounds}, or $\hat\gamma_n(i,m)=0$ if $i=m$;
\item any other medoid $m'\in\mathcal M\setminus G_{Z_i}$ satisfies $Z_{m'}\neq Z_i$ and $m'\neq i$, hence $\hat\gamma_n(i,m')\ge v>u$ by~\eqref{eq:pam-bounds} and~\eqref{eq:pam-gap}.
\end{itemize}
Consequently, the closest medoid of $i$ is the one in $G_{Z_i}$, so the assignment partition $\widehat{\mathcal P}^{\mathrm{PAM}}_{N,n}$ coincides with $\mathcal P^\star$. Therefore
\[
\begin{aligned}
\PP\bigl(\widehat{\mathcal P}^{\mathrm{PAM}}_{N,n}\neq\mathcal P^\star\bigr)
&\le \PP\bigl(\mathcal E_{N,n}(\varepsilon_n)^c\bigr)+\PP\bigl(\mathcal W_N(\delta_0)^c\bigr)\\
&\le N^2\exp\!\left(-\frac{2\varepsilon_n^2n}{C^\star}\right)+2K\exp\bigl(-2N\delta_0^2\bigr)\\
&=N^2\exp\!\left(-\frac{2\left(r-B\sqrt{\frac{\log n}{n}}\right)^2n}{C^\star}\right)+2K\exp\bigl(-2N\delta_0^2\bigr),
\end{aligned}
\]
by Lemma~\ref{lem:uniform} and~\eqref{eq:Wbound}.

Under Assumption~\ref{ass:growth}, the two terms in the preceding bound tend to $0$, and $\PP\bigl(\widehat{\mathcal P}^{\mathrm{PAM}}_{N,n}=\mathcal P^\star\bigr)\longrightarrow1$.
\end{proof}

\begin{proof}[Proof of Theorem~\ref{thm:K-selection}]
Let
\[
x_{N,n}:=\frac{\log N}{n},
\qquad
\varepsilon_{N,n}:=x_{N,n}^{3/8},
\qquad
r_n:=B\sqrt{\frac{\log n}{n}},
\qquad
\beta_{N,n}:=\eta-2r_n-2\varepsilon_{N,n}.
\]
Work on the event $\mathcal E_{N,n}(\varepsilon_{N,n})$ and suppose that $\min_k|G_k|\ge2$. On this event:

\begin{itemize}[leftmargin=*,topsep=2pt]
\item If $Z_i=Z_j=k$, then for every $\ell\notin\{i,j\}$, writing $q:=Z_\ell$, both $\hat\gamma_n(i,\ell)$ and $\hat\gamma_n(j,\ell)$ are distant from at most $\varepsilon_{N,n}$ of the same finite-horizon mean $\Gamma^{(n)}_{kq}$. Hence
\[
\left|\hat\gamma_n(i,\ell)-\hat\gamma_n(j,\ell)\right|
\le2\varepsilon_{N,n},
\]
and $\rho_{N,n}(i,j)\le2\varepsilon_{N,n}$.

\item If $Z_i=k\neq k'=Z_j$, choose $\ell\in G_k\setminus\{i\}$, which is possible since $|G_k|\ge2$. The bounds established in the proof of Theorem~\ref{thm:SL} give
$\hat\gamma_n(i,\ell)\le\gin+r_n+\varepsilon_{N,n},  \hat\gamma_n(j,\ell)\ge\gout-r_n-\varepsilon_{N,n}$.

Hence $\rho_{N,n}(i,j) \ge \eta-2r_n-2\varepsilon_{N,n} = \beta_{N,n}$.

\end{itemize}

Let $h_1^\rho\le\cdots\le h_{N-1}^\rho$ be the single-linkage merge heights associated with $\rho_{N,n}$. Since every true class is a clique at level $2\varepsilon_{N,n}$, at least $N-K$ merges have occurred by that level, and hence
\[
h_{N-K}^\rho\le2\varepsilon_{N,n}.
\]
Conversely, below level $\beta_{N,n}$ there is no edge between two distinct true classes, so at most $N-K$ merges can have occurred. Hence $h_{N-K+1}^\rho\ge\beta_{N,n}$.

Since $\min_k|G_k|\ge2$, we have $K\le N/2$, and $\displaystyle{\left\lceil\frac{N-1}{2}\right\rceil\le N-K}$.

It follows that $h_{\mathrm{med}}^\rho\le2\varepsilon_{N,n},  h_{\max}^\rho\ge\beta_{N,n}$.

Moreover, since $0\le\hat\gamma_n(i,j)\le M$, we have $\rho_{N,n}(i,j)\le M$ for every $i\neq j$, and hence $h_{\max}^\rho\le M$.

We now control the data-driven threshold. By its definition,
\[
a_{N,n}
\ge
h_{\max}^\rho x_{N,n}^{1/4}
\ge
\beta_{N,n}x_{N,n}^{1/4}
>
2\varepsilon_{N,n}.
\]
On the other hand,
\[
\sqrt{h_{\mathrm{med}}^\rho h_{\max}^\rho}
\le
\sqrt{2M\varepsilon_{N,n}}
<
\beta_{N,n},
\]
and
\[
h_{\max}^\rho x_{N,n}^{1/4}
\le
Mx_{N,n}^{1/4}
<
\beta_{N,n}.
\]
Thus
\[
2\varepsilon_{N,n}<a_{N,n}<\beta_{N,n}.
\]

It follows from the two profile bounds above that two vertices are adjacent in $H_{N,n}$ if and only if they belong to the same true cluster. Hence $H_{N,n}$ has exactly $K$ connected components and $\hat K_{N,n}=K$. Consequently,
\[
\PP(\hat K_{N,n}\neq K)
\le
\PP\bigl(\mathcal E_{N,n}(\varepsilon_{N,n})^c\bigr)
+
\PP\left(\min_k|G_k|<2\right).
\]

By Lemma~\ref{lem:uniform},
\[
\PP\bigl(\mathcal E_{N,n}(\varepsilon_{N,n})^c\bigr)
\le
N^2
\exp\!\left(
-\frac{2\varepsilon_{N,n}^2n}{C^\star}
\right).
\]
Moreover,
\[
\begin{aligned}
\PP\left(\min_k|G_k|<2\right)
&\le
\sum_{k=1}^K
\left[
(1-w_k)^N
+
Nw_k(1-w_k)^{N-1}
\right]\\
&\le
K\left[
(1-w_{\min})^N
+
N(1-w_{\min})^{N-1}
\right].
\end{aligned}
\]
This proves the finite-sample bound.

Finally, under Assumption~\ref{ass:growth}, $x_{N,n}\to0$, and the two terms in the bound tend to zero. Hence $\PP(\hat K_{N,n}=K)\longrightarrow1$.

\end{proof}

\subsection{Proofs for Section~\ref{sec:extensions}}

\begin{proof}[Proof of Proposition~\ref{prop:gen-concentration}]
\emph{(i)} Assume first that the hidden chains are stationary. Then each observed process is stationary and ergodic, and the two observed processes are independent. Step~1 of the proof of Proposition~\ref{prop:convergence} applies verbatim, with $M^{\mathrm{mc}}$ in place of $M$, and gives the almost sure convergence to the stated deterministic limit. For arbitrary initial laws, Step~2 of the proof of Proposition~\ref{prop:convergence} applies to the hidden chains, whose path laws are absolutely continuous with respect to their stationary counterparts. Adding the independent emission noises preserves this absolute continuity, and so does passing to the observed processes through the emission maps. Hence the same almost sure limit holds.

\emph{(ii)} Write $\widetilde Z_t:=(Z^{(k)}_t,Z^{(k')}_t)$ and $\widetilde X_t:=(X^{(k)}_t,X^{(k')}_t)$. Then $(\widetilde X_t)$ is a hidden Markov process with underlying chain $(\widetilde Z_t)$ in the sense of \cite[Example~2.15]{paulin}. For
\[
f(\widetilde X_{1:n}):=\frac1n\dom(X^{(k)}_{1:n},X^{(k')}_{1:n}),
\]
changing one observed coordinate from $(x,y)$ to $(x',y')$ changes $f$ by at most
\[
\frac{c^{\mathrm{mc}}_{\sub}(x,x')+c^{\mathrm{mc}}_{\sub}(y,y')}{n}\le\frac{2M^{\mathrm{mc}}}{n}.
\]
Applying \cite[Corollary~2.16]{paulin} to these bounded differences, with the mixing time of $R_k\otimes R_{k'}$, gives the stated concentration inequality.

\emph{(iii)} For $k=k'$, regard each multichannel observation as a single letter of the finite product alphabet $\Sigma=\Sigma_1\times\cdots\times\Sigma_J$, and define
\[
s(a,b):=\delta^{\mathrm{mc}}(a)+\delta^{\mathrm{mc}}(b)-c^{\mathrm{mc}}_{\sub}(a,b)
=\sum_{j=1}^J\lambda_j\bigl[\delta^{(j)}(a_j)+\delta^{(j)}(b_j)-c^{(j)}_{\sub}(a_j,b_j)\bigr].
\]
Thus Lemma~\ref{lem:OM-score-properties} applies unchanged on $\Sigma$. The induced alignment score $L_s$ is superadditive, satisfies the rectangular-block comparison, is linearly bounded, and has bounded differences. The partition argument of \cite[Theorem~3.1]{houdre2019}, specialized to two independent HMMs with identical parameters as in \cite[Corollary~3.1]{houdre2019}, applies to $L_s$, using the hidden-Markov concentration inequality of \cite[Corollary~2.16]{paulin}. Together with~\eqref{eq:OM-score-identity}, this gives the stated $O(\sqrt{\log n/n})$ rate in the stationary case. For arbitrary common initial laws, the coupling argument of Proposition~\ref{prop:diagonal-rate} applies to the hidden chains, using common emission noises after coupling, and gives the same rate after adjusting the constant.

\end{proof}

\begin{proof}[Proof of Theorem~\ref{thm:gen-clustering}]
Set $C^{\mathrm{mc},\star}:=\max_{k,k'}C^{\mathrm{mc}}_{kk'}$. Proposition~\ref{prop:gen-concentration} (ii) and a union bound give the analogue of Lemma~\ref{lem:uniform}. Moreover, Proposition~\ref{prop:gen-concentration} (iii), together with the coupling argument used in the proof of Theorem~\ref{thm:SL}, gives a constant $B^{\mathrm{mc}} \in \R_+$ such that the same within- and between-component bounds hold with $\Gamma$, $B$ and $C^\star$ replaced by $\Gamma^{\mathrm{mc}}$, $B^{\mathrm{mc}}$ and $C^{\mathrm{mc},\star}$.

The proofs of Theorem~\ref{thm:SL}, Lemma~\ref{lem:kmedoid-one-per-cluster} and Theorem~\ref{thm:kmedoids} are deterministic and apply verbatim. In particular, the cancellations of equal finite-horizon means used in the PAM and profile arguments are unchanged. The proof of Theorem~\ref{thm:K-selection} applies with $M^{\mathrm{mc}}$ in place of $M$ in the corresponding bounds, which gives the result for $\hat K_{N,n}$.
\end{proof}
\appendix
\section{Additional Figures}
\label{app:figures}

Figures~\ref{fig:gamma-convergence-trate} and~\ref{fig:gamma-convergence-random} give the supplementary convergence results for the TRATE and random cost schemes. Figures~\ref{fig:ARIconvergencesingle} and~\ref{fig:ARIgridsingle} report the single-linkage results. Figures~\ref{fig:hmm-gamma-convergence}--\ref{fig:hmm-k-selection} report the multichannel HMM experiments described in Section~\ref{ssec:num-hmm}.

\begin{figure}[tbp]
\centering
\includegraphics[width=0.7\linewidth]{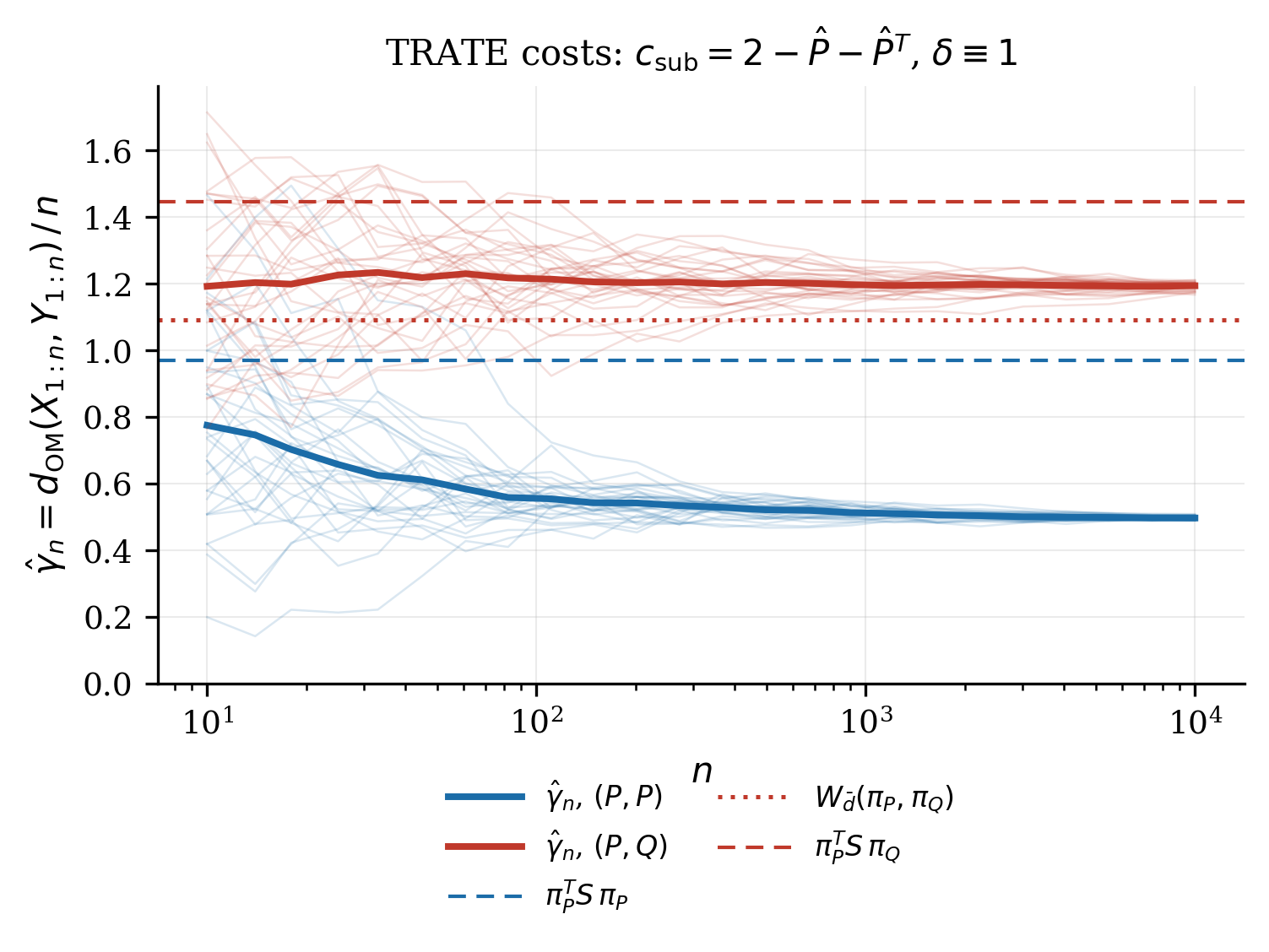}
\caption{Sample paths of $\hat\gamma_n$ for the within (blue) and between (red) configurations, with the bounds of Proposition~\ref{prop:bounds} as horizontal references. TRATE cost scheme, $30$ replicates.}
\label{fig:gamma-convergence-trate}
\end{figure}

\begin{figure}[tbp]
\centering
\includegraphics[width=0.7\linewidth]{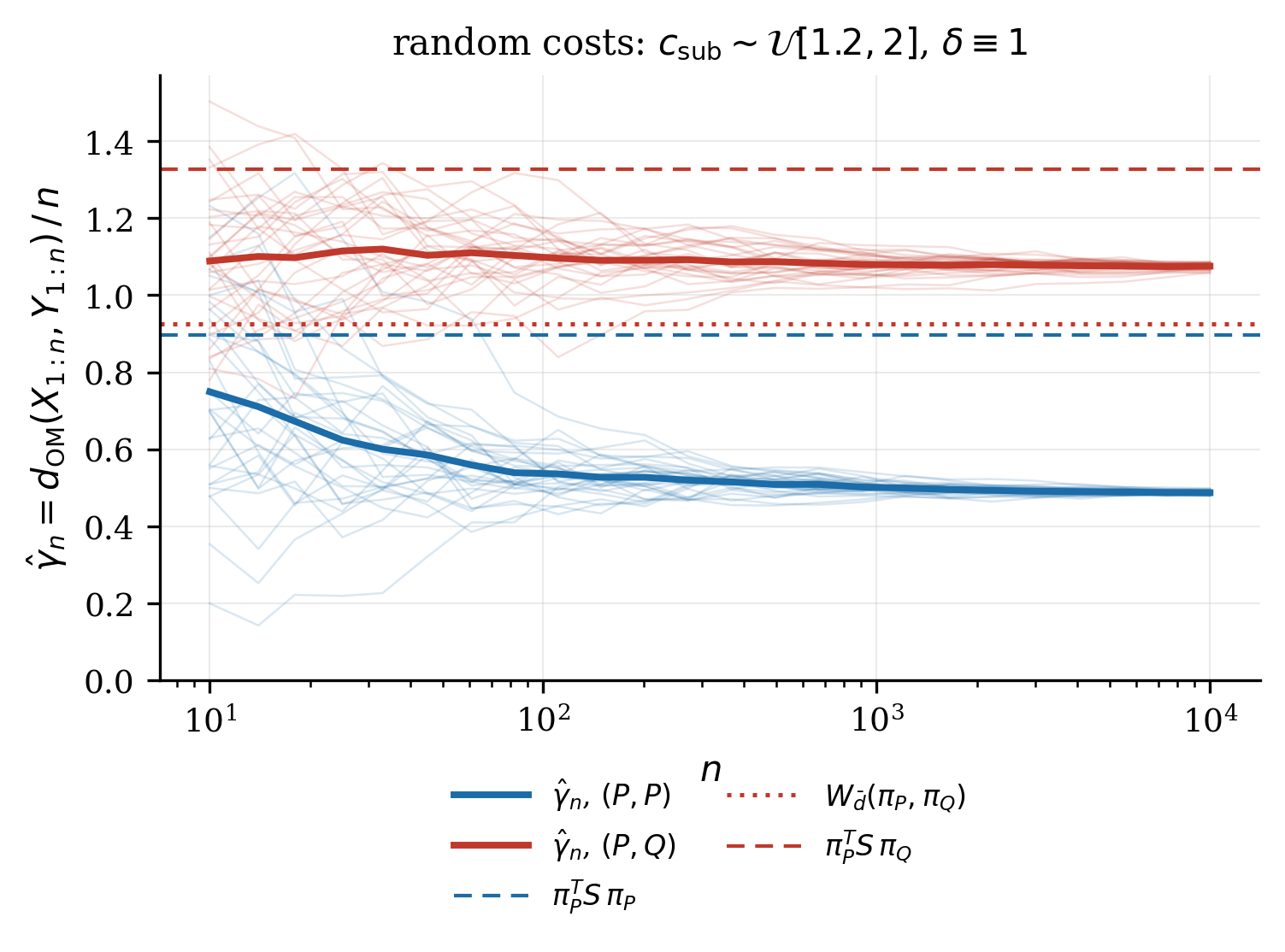}
\caption{Sample paths of $\hat\gamma_n$ for the within (blue) and between (red) configurations, with the bounds of Proposition~\ref{prop:bounds} as horizontal references. Random symmetric cost scheme, $30$ replicates.}
\label{fig:gamma-convergence-random}
\end{figure}

\begin{figure}[tbp]
\centering
\includegraphics[width=0.7\linewidth]{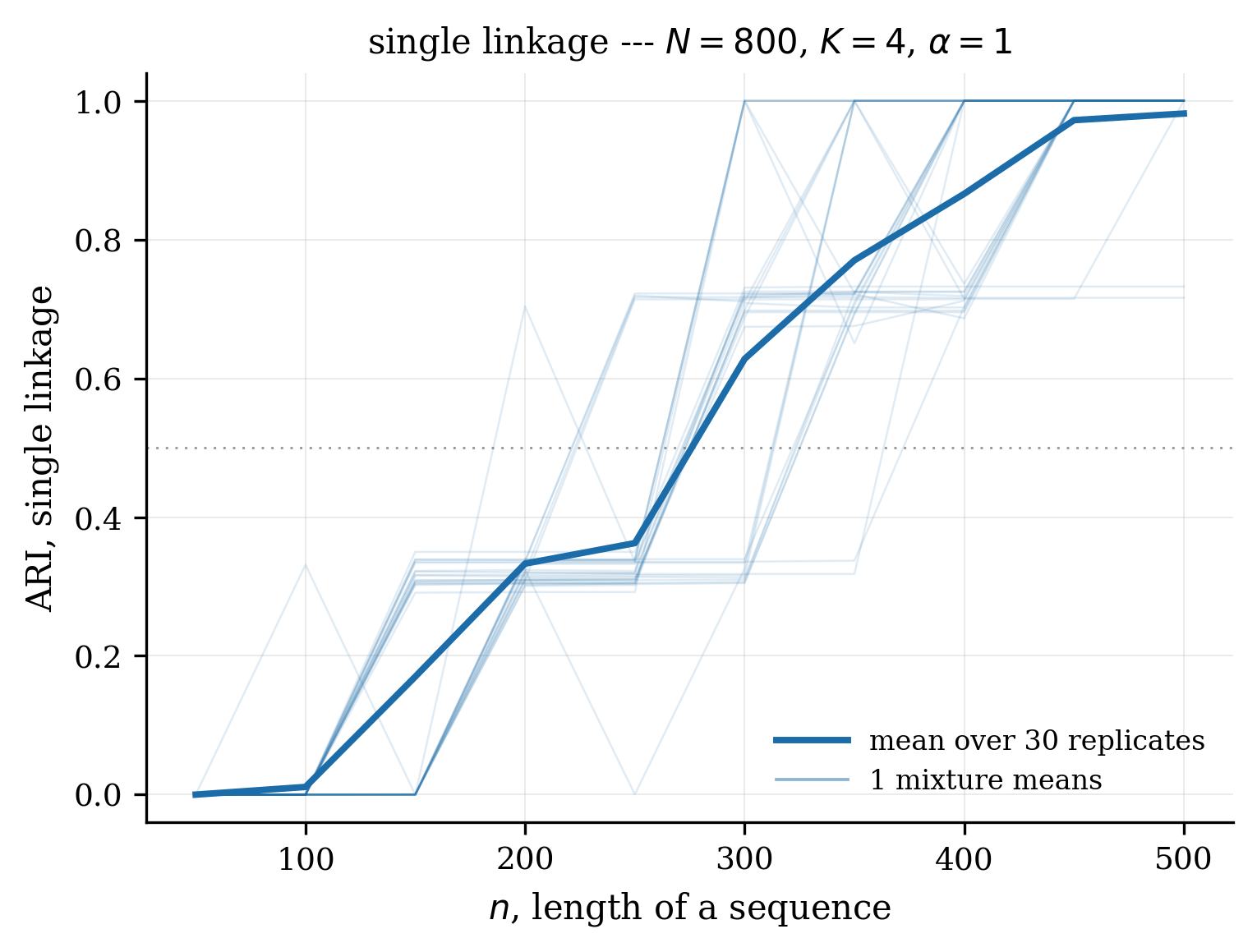}
\caption{$\mathrm{ARI}$ trajectories in function of $n$ for fixed $N=800$, $K=4$ and $\alpha=1$, with $R=50$ replications. Single linkage.}
\label{fig:ARIconvergencesingle}
\end{figure}

\begin{figure}[tbp]
\centering
\includegraphics[width=0.9\linewidth]{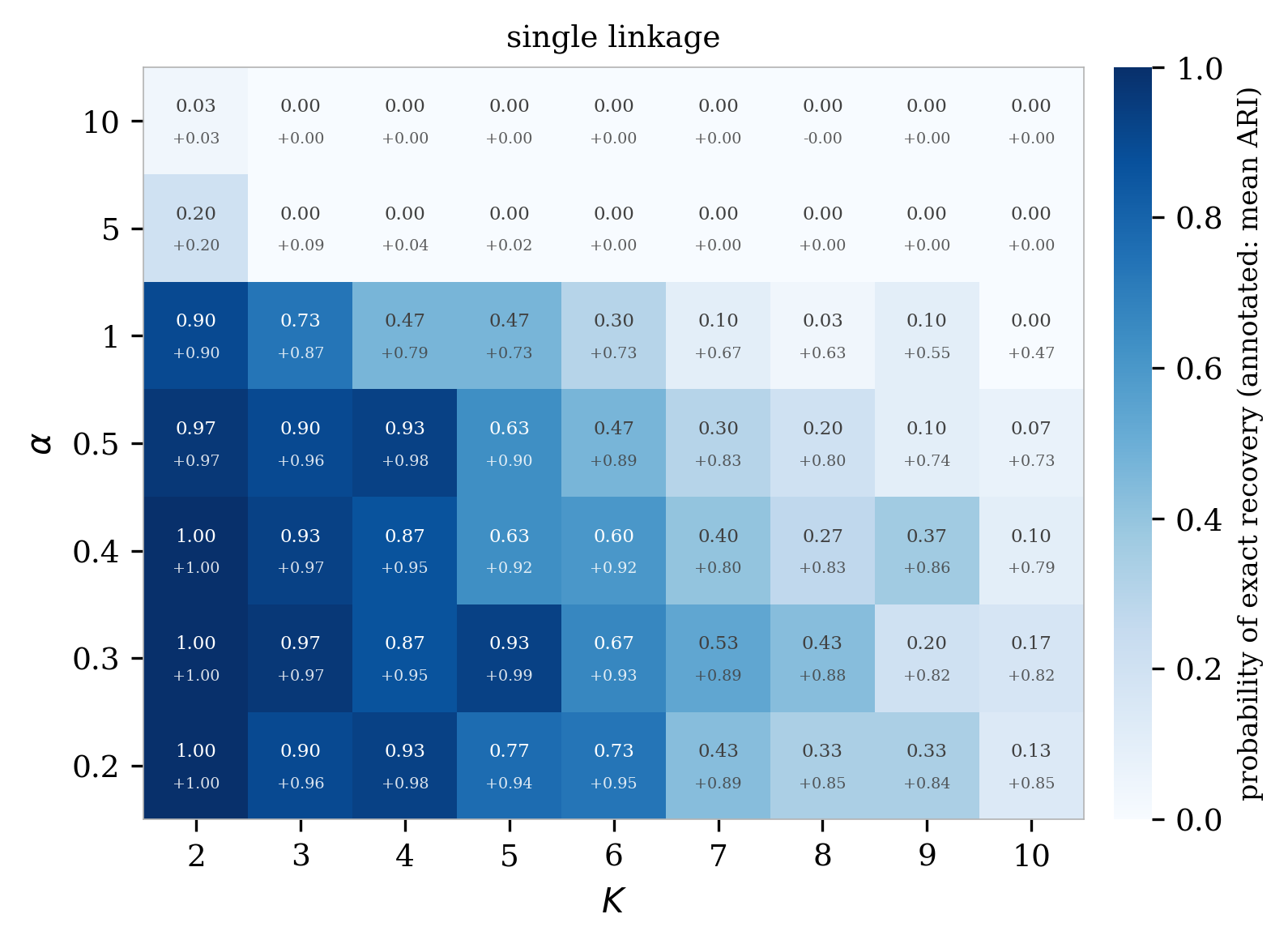}
\caption{Proportion of exact recoveries (above) and mean \textit{ARI}s (below), as functions of $\alpha$ (rows) and $K$ (columns), for Single Linkage.}
\label{fig:ARIgridsingle}
\end{figure}

\begin{figure}[tbp]
\centering
\begin{minipage}[t]{0.32\linewidth}
  \centering
  \includegraphics[width=\linewidth]{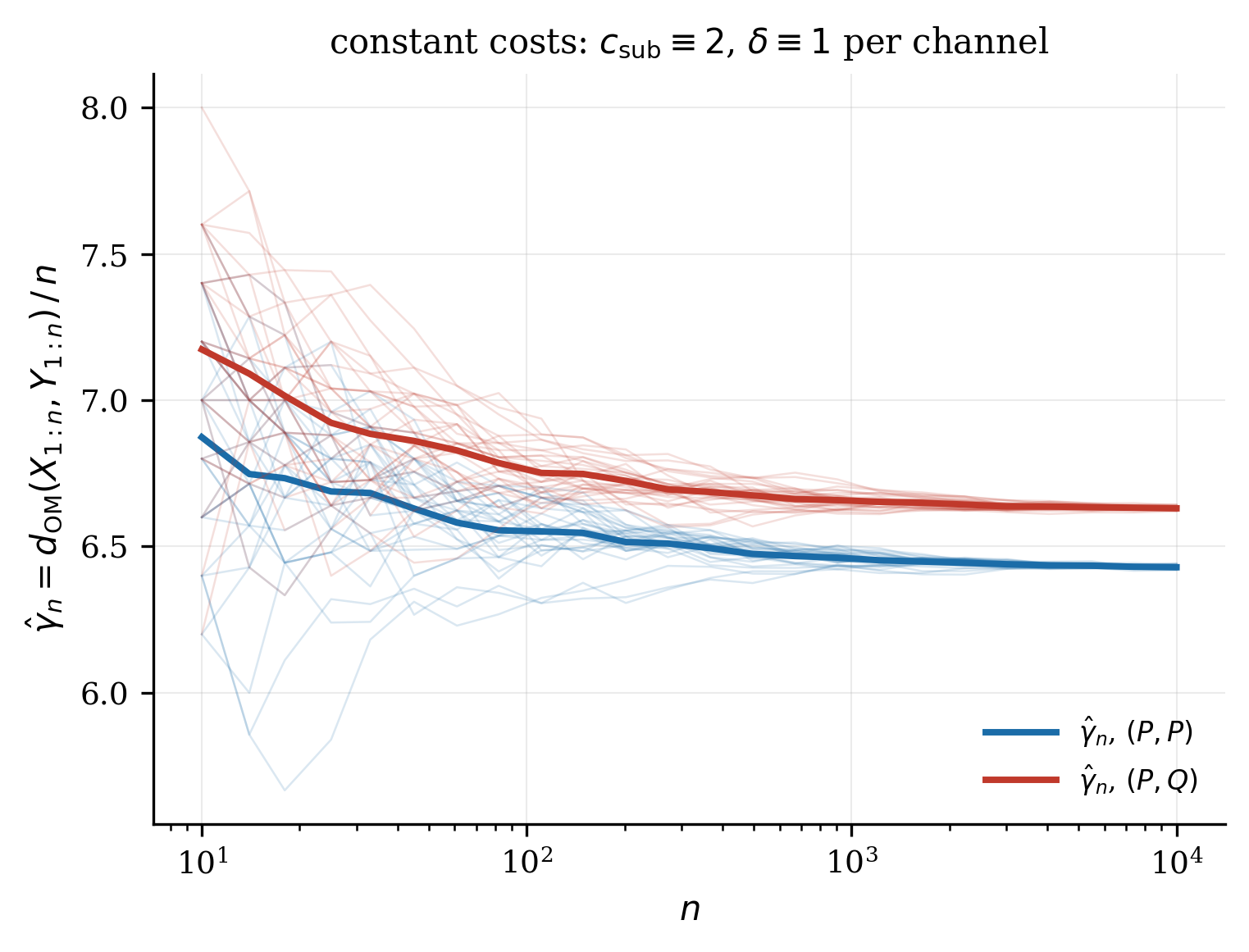}

  (a) Constant costs.
\end{minipage}
\hfill
\begin{minipage}[t]{0.32\linewidth}
  \centering
  \includegraphics[width=\linewidth]{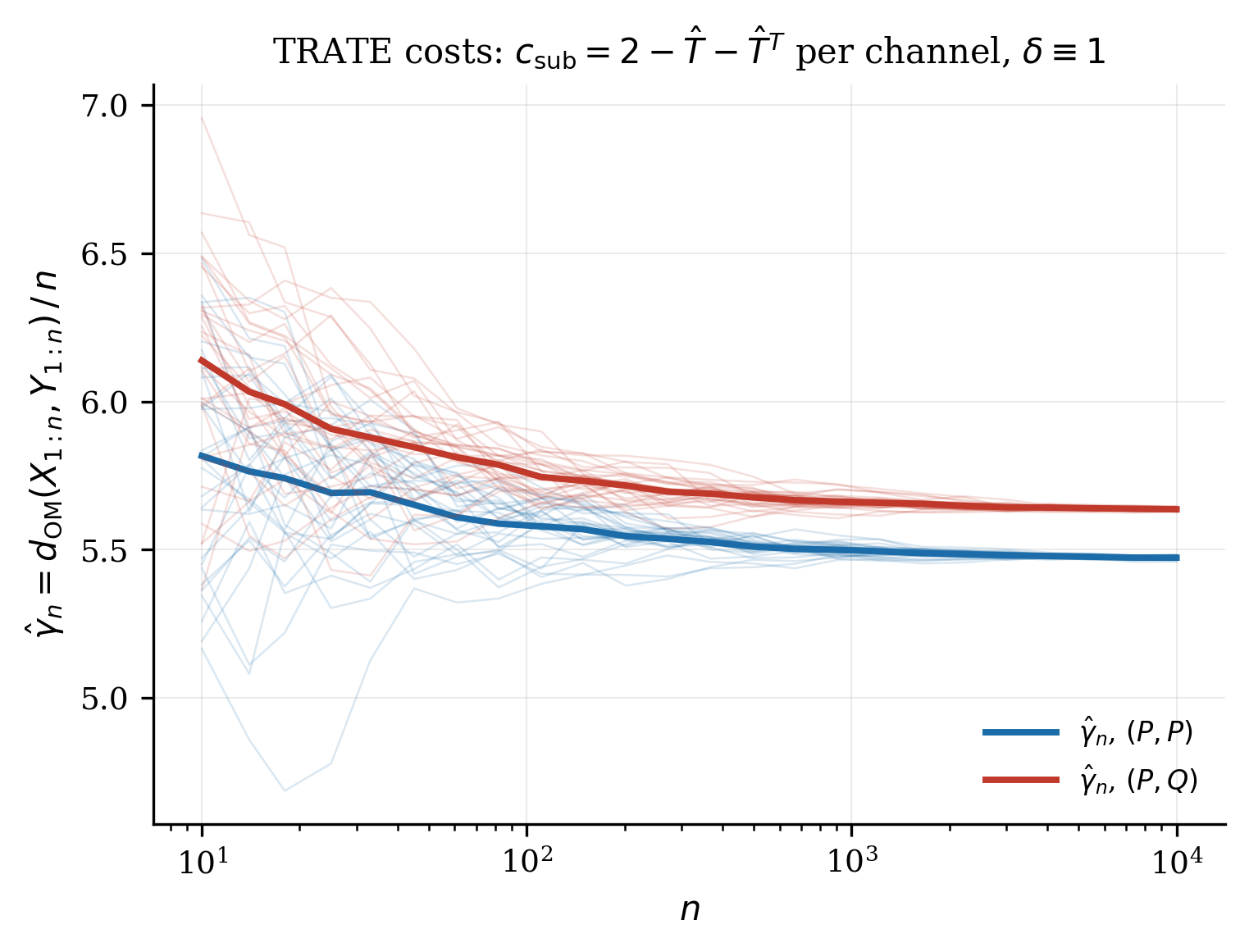}

  (b) TRATE costs.
\end{minipage}
\hfill
\begin{minipage}[t]{0.32\linewidth}
  \centering
  \includegraphics[width=\linewidth]{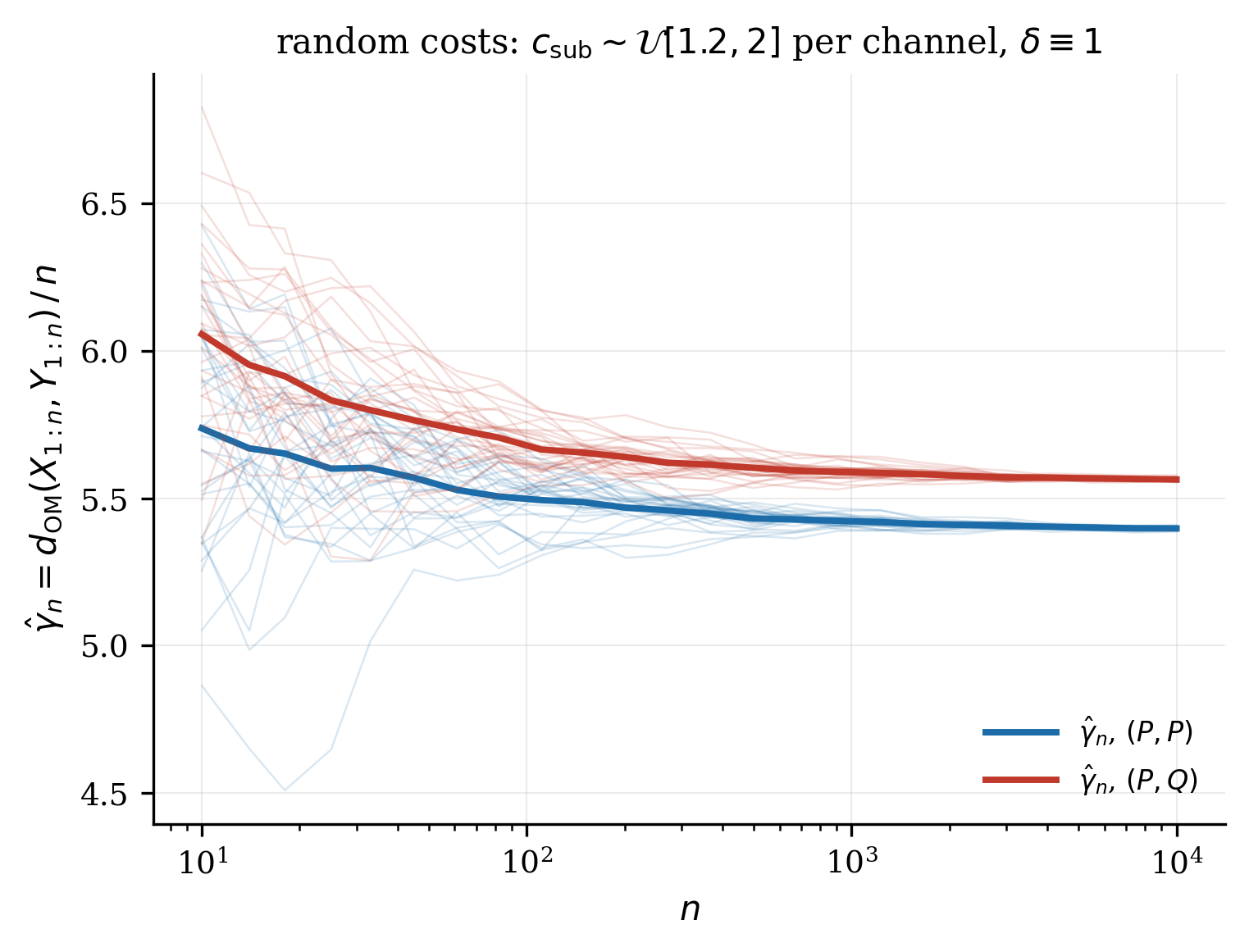}

  (c) Random costs.
\end{minipage}
\caption{Sample paths of $\hat\gamma_n^{\mathrm{mc}}$ for the within (blue) and between (red) configurations, with the corresponding bounds as horizontal references. Multichannel HMMs, $30$ replicates.}
\label{fig:hmm-gamma-convergence}
\end{figure}

\begin{figure}[tbp]
\centering
\includegraphics[width=\linewidth]{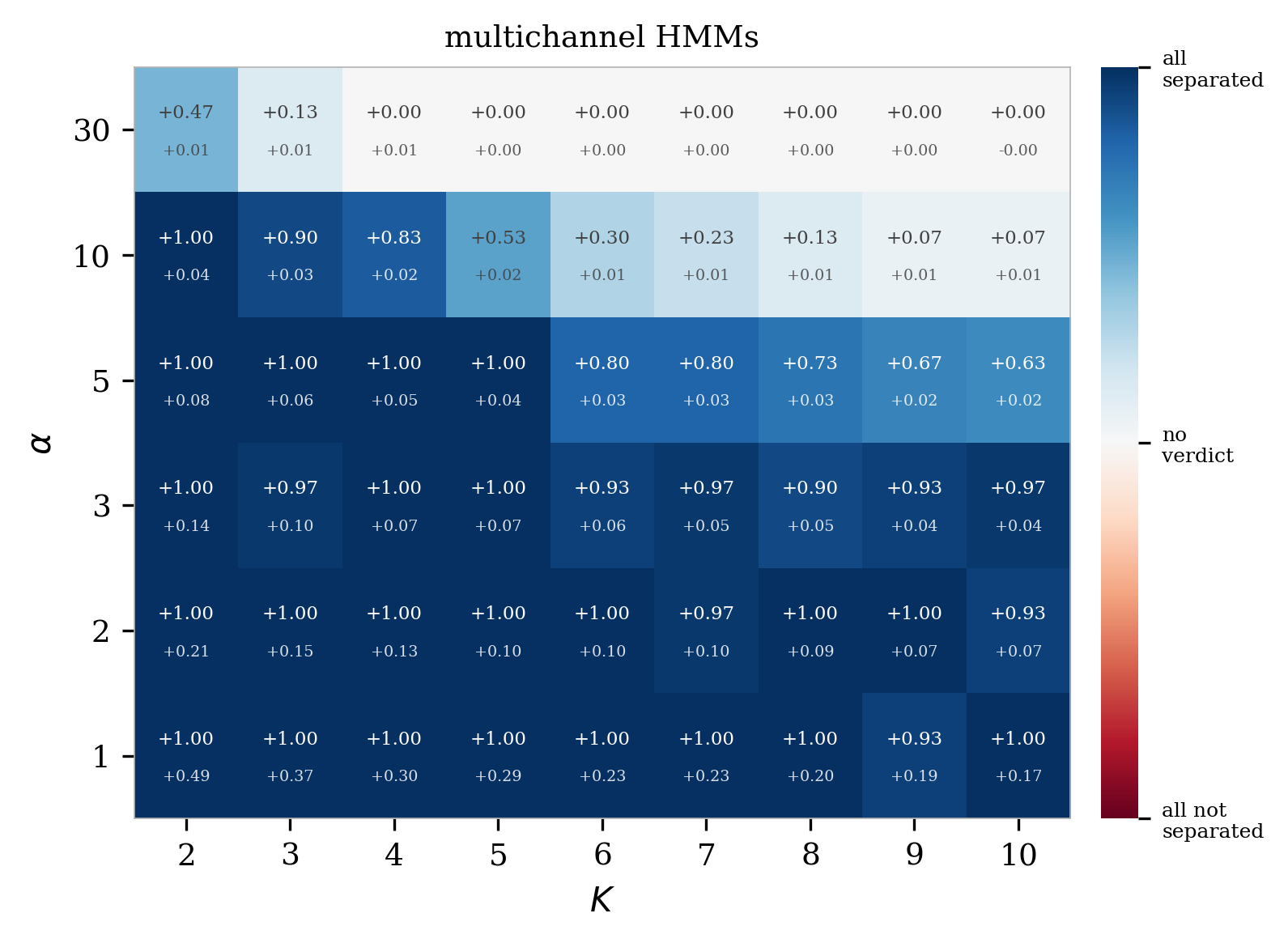}
\caption{Heatmap of the verdict balance $\PP(\text{separated})-\PP(\text{nonseparated})$ at level $0.95$, as a function of $\alpha$ (rows) and $K$ (columns), with signed median $\hat\eta_n^{\mathrm{mc}}$ (below) in each cell. Multichannel HMMs.}
\label{fig:hmm-separation}
\end{figure}

\begin{figure}[tbp]
\centering
\begin{minipage}[t]{0.32\linewidth}
  \centering
  \includegraphics[width=\linewidth]{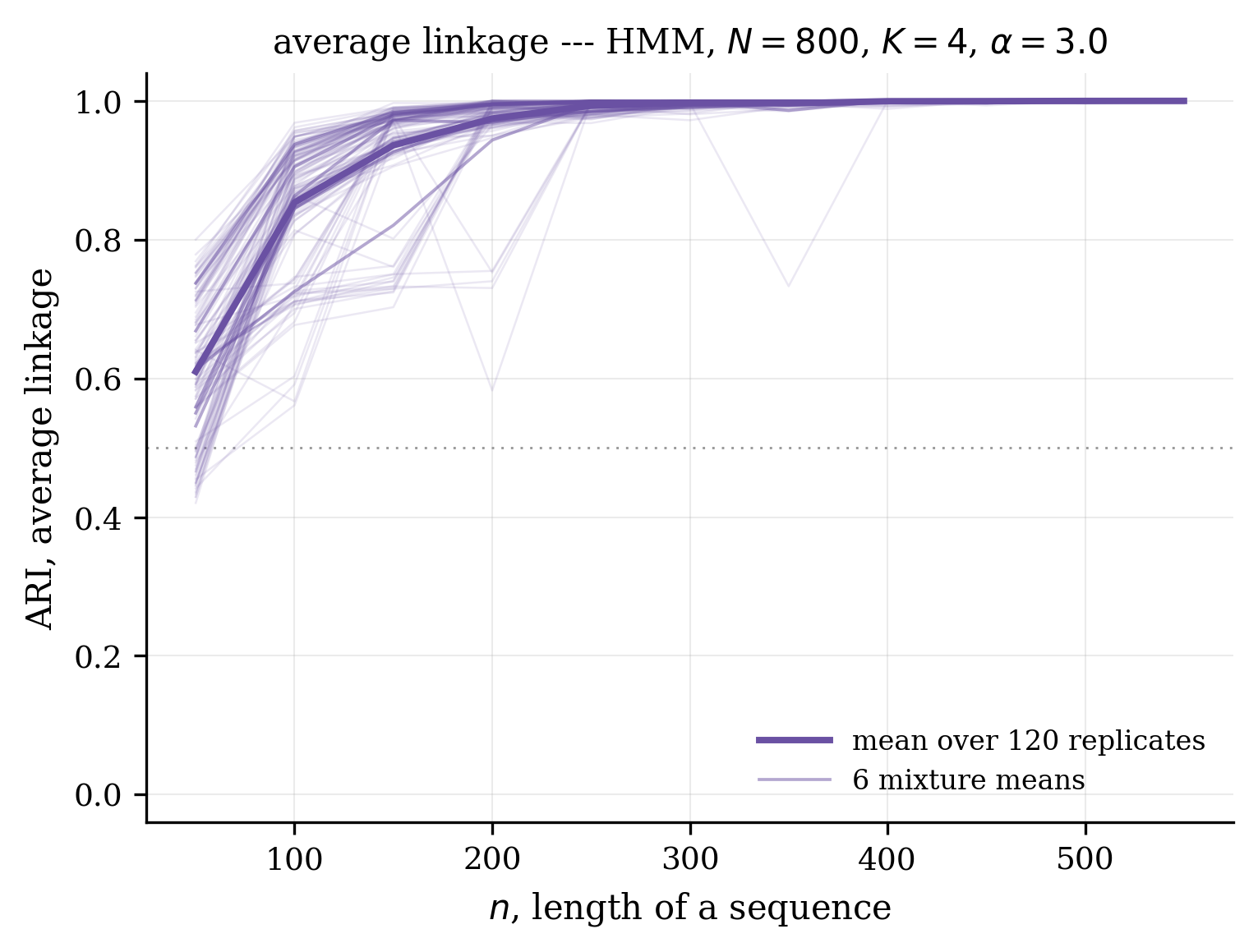}

  (a) Average linkage.
\end{minipage}
\hfill
\begin{minipage}[t]{0.32\linewidth}
  \centering
  \includegraphics[width=\linewidth]{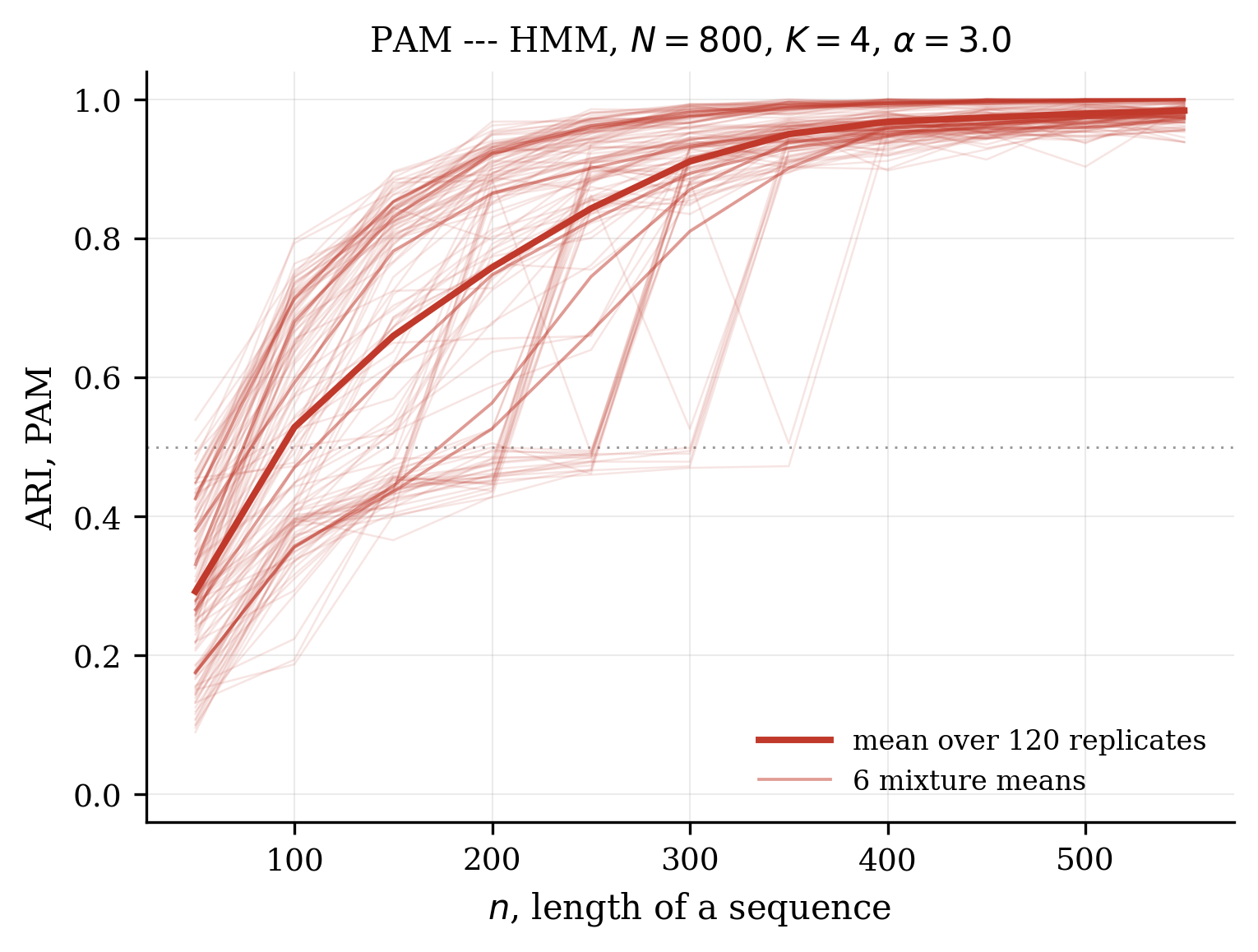}

  (b) K-medoids (PAM).
\end{minipage}
\hfill
\begin{minipage}[t]{0.32\linewidth}
  \centering
  \includegraphics[width=\linewidth]{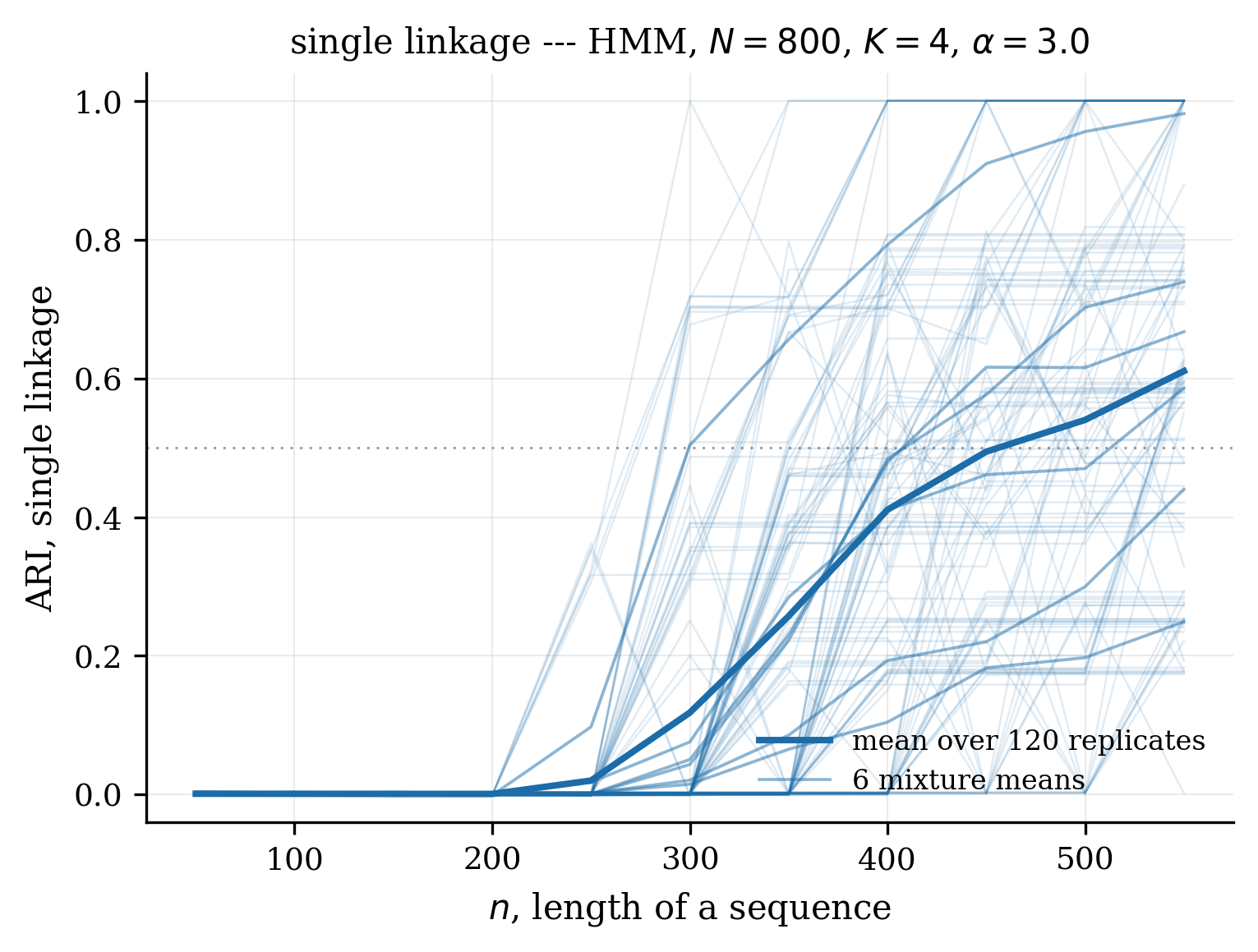}

  (c) Single linkage.
\end{minipage}
\caption{$\mathrm{ARI}$ trajectories in function of $n$ for fixed $N=800$, $K=4$ and $\alpha=3$, with $R=50$ replications. Multichannel HMMs.}
\label{fig:hmm-ari-path}
\end{figure}

\begin{figure}[tbp]
\centering
\includegraphics[width=\linewidth]{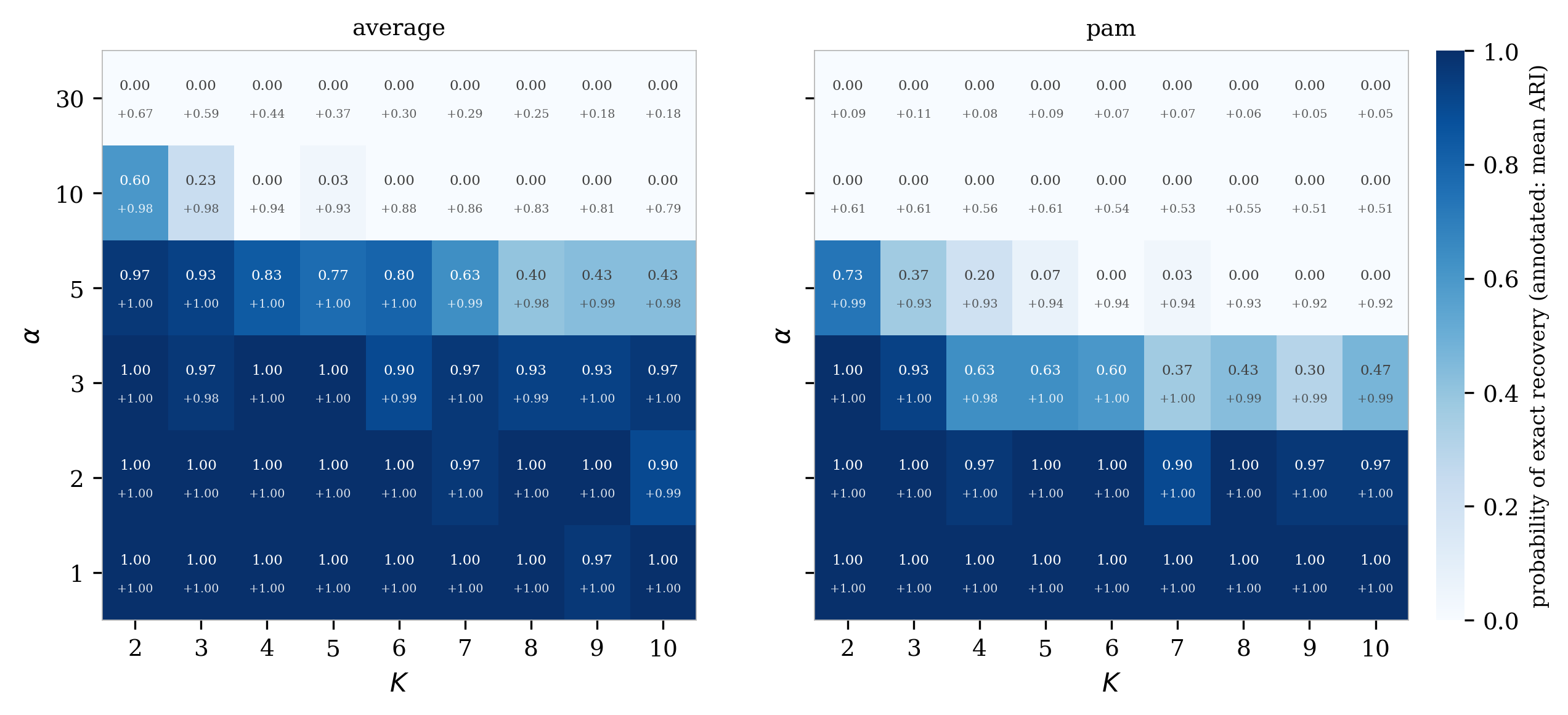}
\caption{Proportion of exact recoveries (above) and mean \textit{ARI}s (below), as functions of $\alpha$ (rows) and $K$ (columns), for Average Linkage and PAM. Multichannel HMMs.}
\label{fig:hmm-recovery}
\end{figure}

\begin{figure}[tbp]
\centering
\includegraphics[width=\linewidth]{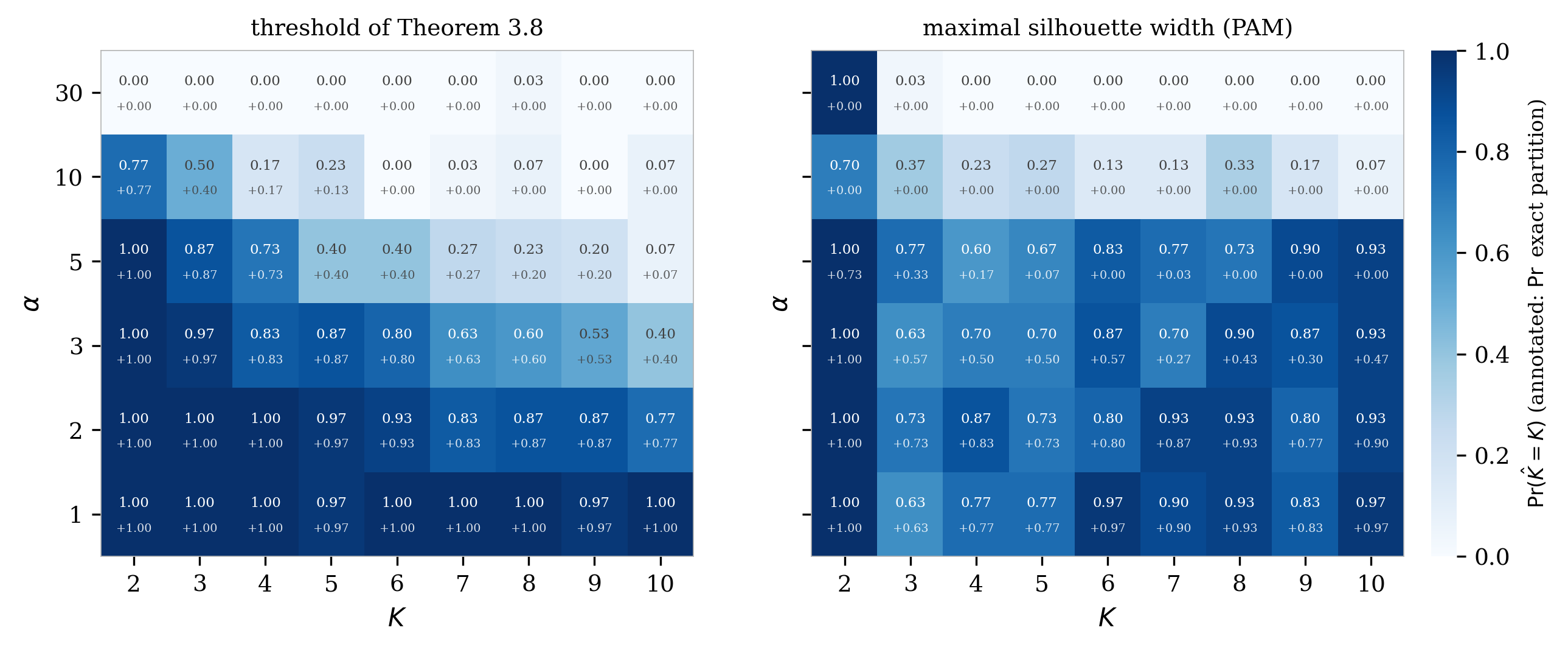}
\caption{Proportion of repetitions with $\hat{K}=K$ using the rule of Theorem~\ref{thm:K-selection} and Average Silhouette Width as functions of $\alpha$ and $K$. Multichannel HMMs.}
\label{fig:hmm-k-selection}
\end{figure}

\newpage
\bibliographystyle{unsrt}
\bibliography{bibliographie}

\end{document}